\documentclass[11pt,twoside, leqno]{article}
\usepackage{mathrsfs,amssymb,amsmath,amsthm,amsfonts}
\usepackage{graphicx}
\usepackage{subfig}
\usepackage{float}

\usepackage{color}
\usepackage[colorlinks, linkcolor=red, anchorcolor=green, citecolor=blue]{hyperref}

\allowdisplaybreaks

\newtheorem{theorem}{Theorem}[section]
\newtheorem{lemma}[theorem]{Lemma}
\newtheorem{corollary}[theorem]{Corollary}
\newtheorem{proposition}[theorem]{Proposition}
\newtheorem{remark}[theorem]{Remark}
\newtheorem{definition}[theorem]{Definition}

\numberwithin{equation}{section}

\begin{document}

\title{\Large\bf Matrix Recovery from Rank-One Projection Measurements via Nonconvex Minimization \footnotetext{\hspace{-0.35cm}
\endgraf $^\ast$\,Corresponding author.
\endgraf {1.P. Li is with  Graduate School, China Academy of Engineering Physics, Beijing 100088, China (E-mail: lipeng16@gscaep.ac.cn)}
\endgraf{2.W. Chen is with Institute of Applied Physics and Computational Mathematics, Beijing 100088, China (E-mail: chenwg@iapcm.ac.cn)}
}}
\author{Peng Li$^{1*}$ and Wengu Chen$^2$}
\date{ }

\maketitle
\textbf{Abstract.} In this paper, we consider the matrix recovery from rank-one projection measurements proposed in [Cai and Zhang,  Ann. Statist., 43(2015), 102-138], via nonconvex minimization. We establish a sufficient identifiability condition, which can guarantee the exact recovery of low-rank matrix via Schatten-$p$ minimization $\min_{X}\|X\|_{S_p}^p$ for $0<p<1$ under affine constraint, and stable recovery of low-rank matrix under $\ell_q$ constraint and Dantzig selector constraint. Our condition is also sufficient to guarantee low-rank matrix recovery via least $q$ minimization $\min_{X}\|\mathcal{A}(X)-b\|_{q}^q$ for $0<q\leq1$. And we also extend our result to Gaussian design distribution, and show that any matrix can be stably recovered for rank-one projection from Gaussian distributions via least $1$ minimization with high probability.

\textbf{Key Words and Phrases.} Low-rank matrix recovery, Rank-one projection, $\ell_q$-Restricted uniform boundedness, Schatten-$p$ minimization, Least $q$ minimization.

\textbf{MSC 2010}. {62H10, 62H12,90C26, 94A08}






\section{Introduction}\label{s1}

\hskip\parindent

As is well known to us, a closely related problem to compressed sensing, which was initiated by Cand\`{e}s, Romberg and Tao's seminal works \cite{CRT2006,CRT2006-1} and Donoho's ground-breaking work \cite{D2006}, is the low-rank matrix recovery.
It aims to recover an unknown low-rank matrix based on
its affine transformation
\begin{align}\label{Matrixsystemequationsnoise}
b=\mathcal{A}(X)+z,
\end{align}
where $X\in\mathbb{R}^{m\times n}$ is the decision variable and the linear map $\mathcal{A}: \mathbb{R}^{m\times n}\rightarrow \mathbb{R}^L$, $z\in\mathbb{R}^L$ is a measurement error and $b\in\mathbb{R}^L$ is measurements. The linear map $\mathcal{A}$ can be equivalently specified by $L$ $m\times n$ measurement matrices $A_1,\ldots,A_L$ with
\begin{align}\label{Matrixsystemequations}
[\mathcal{A}(X)]_j=\langle A_j,X\rangle,
\end{align}
where the inner product of two matrices of the same dimensions is defined as
$\langle X, Y\rangle=\sum_{i,j}X_{ij}Y_{ij}=\text{trace}(X^{T}Y)$.
Low-rank matrices arise in an incredibly wide range of settings throughout science and applied mathematics. To name just a few examples, we commonly encounter low-rank matrices in contexts as varied as: ensembles
of signals \cite{DE2012,AR2015}, system identification \cite{LV2009}, adjacency matrices \cite{LLR1995}, distance matrices \cite{BLWY2006,BG2010,SY2007}, machine learning \cite{AEP2008,OTJ2010,D2004}, and quantum state tomography \cite{GLFBE2010,ABHMM2013}.

Let $\|X\|_{S_p}:=\|\sigma(X)\|_p=\big(\sum_{j=1}^{\min\{m,n\}}|\sigma_j|^p\big)^{1/p}$ denote the Schatten $p$-norm of the matrix $X$, where $\sigma(X)=(\sigma_1,\ldots,\sigma_{\min\{m,n\}})$ is the vector of singular values of the matrix $X$ and  $\sigma_j$ denote the $j$-th largest singular value of $X$. We should point out that $\|X\|_{S_1}=\|X\|_*$ the nuclear norm,
$\|X\|_{S_2}=\|X\|_F$ the Frobenius norm and $\|X\|_{S_{\infty}}=\|X\|$ the operator norm.
In 2010, Recht,  Fazel and Parrilo \cite{RFP2010} generalized the restricted isometry property (RIP) in \cite{CT2005,D2006} from vectors to matrices and showed that
if certain restricted isometry property holds for the linear transformation $\mathcal{A}$, then the low-rank solution can be recovered by solving the the nuclear norm minimization problem
\begin{align*}
\min_{X}\|X\|_{*}~~\text{subject~to~}\mathcal{A}(X)=b.
\end{align*}
Later, the low-rank matrix recovery problem has been studied by many scholars, readers can refer to \cite{CP2011,CZ2014,CZ2015,ZL2017} under matrix restricted isometry property, and \cite{FR2013,KKRT2016} under rank null space property.

In this paper, we consider one matrix recovery model with additional structural assumption-the rank-one projection (ROP), which introduced by Cai and Zhang \cite{CZ2015}. Under the ROP model, we observe
\begin{align}\label{ROP1}
b_{j}=(\beta^j)^{T}X\gamma^j+z_j,~~j=1,\ldots,L,
\end{align}
where $\beta^j$ and $\gamma^j$ are random vectors with entries independently drawn from some distribution $\mathcal{P}$, $z_j$ are random errors. In view of the linear map in (\ref{Matrixsystemequations}), it can be rewrite as
\begin{align}\label{ROP2}
b_{j}=[\mathcal{A}(X)]_j+z_j=\langle \beta^j(\gamma^j)^{T},X\rangle +z_j,~~j=1,\ldots,L,
\end{align}
i.e. $A_j=\beta^j(\gamma^j)^{T}$ for $j=1,\ldots,L$. For this model, Cai and Zhang proposed a constrained nuclear norm minimization method, which can be stated as follows
\begin{align}\label{MatrixS1}
\min_{X}~\|X\|_{S_1} ~~\text{subject~ to}~ ~b-\mathcal{A}(X)\in\mathcal{B}.
\end{align}
In the noiseless case, they took $\mathcal{B}=\{0\}$, i.e.,
\begin{align}\label{MatrixS1-Noiseless}
\min_{X}~\|X\|_{S_1} ~~\text{subject~ to}~ ~\mathcal{A}(X)=b.
\end{align}
And in the noisy case, they took $\mathcal{B}=\mathcal{B}^{\ell_1}(\eta_1)\cap\mathcal{B}^{DS}(\eta_2)$, i.e.,
\begin{align}\label{MatrixS1-Noise}
\min_{X}~\|X\|_{S_1} ~~\text{subject~ to}~ ~\|b-\mathcal{A}(X)\|_1/L\leq\eta_1,~~ \|\mathcal{A}^*(b-\mathcal{A}(X))\|_{S_\infty}\leq\eta_2.
\end{align}

However, the matrix RIP framework is not well suited for the ROP model and would lead to suboptimal results.
Cai and Zhang \cite{CZ2015} introduced the restricted uniform boundedness (RUB) condition (see Definition \ref{lqRUB}),
which can guarantee the exact recovery of low-rank matrices in the noiseless case and stable recovery in the noisy case through the constrained nuclear norm minimization. They also showed that the RUB condition are satisfied by
sub-Gaussian random linear maps with high probability.

The ROP model can be further simplified by taking $\beta^j=\gamma^j$ if the low-rank matrix $X$ is known to be symmetric. This can be found in many problem, for example, low-dimensional Euclidean embedding \cite{T2000,RFP2010} , phase retrieval \cite{CESV2013-2015, CSV2013, CL2013, LV2013} and  covariance matrix estimation \cite{CMW2013,CMW2015,CCG2015}. In this case, the ROP design can be simplified to symmetric rank-one projections (SROP)
\begin{align}\label{SROP}
b_{j}=\langle \beta^j(\beta^j)^{T}, X\rangle+z_j,~~j=1,\ldots,L.
\end{align}

Unfortunately, the original sampling operator $\mathcal{A}$ does not satisfy RUB. This occurs primarily because each measurement
matrix $A_i$ has non-zero mean, which biases the output measurements. In order to get rid of this undesired bias
effect, Cai and Zhang \cite{CZ2015} (see also \cite{CCG2015}) introduced a set of ``debiased" auxiliary measurement matrices as follows
\begin{align*}
\tilde{A}_j=A_{2j-1}-A_{2j}, ~j=1,\ldots,\bigg\lfloor\frac{L}{2}\bigg\rfloor.
\end{align*}
By this notation, we can define a linear map $\tilde{\mathcal{A}}:~\mathbb{S}^{m}\rightarrow \mathbb{R}^{\lfloor L/2 \rfloor}$ by
\begin{align}\label{SROP2}
\tilde{b}_j=[\tilde{\mathcal{A}}(X)]_j+\tilde{z}_j=\langle\tilde{A}_j, X\rangle+\tilde{z}_j,
\end{align}
where $\tilde{b}_j=b_{2j-1}-b_{2j}$, $\tilde{z}_j=z_{2j-1}-z_{2j}$ and $\mathbb{S}^{m}$ denotes the set of all $m\times m$ symmetric matrices.
Owing to $\mathcal{A}(X)=b$ implying $\tilde{\mathcal{A}}(X)=b$, they still considered (\ref{MatrixS1-Noiseless}) in noiseless case. And note that $\|\tilde{b}-\tilde{\mathcal{A}}(X)\|_1\leq\|b-\mathcal{A}(X)\|_1$, therefore they consider
\begin{align}\label{SROP-MatrixS1-Noise}
\min_{X}~\|X\|_{S_1} ~~\text{subject~ to}~ ~\|b-\mathcal{A}(X)\|_1/L\leq\eta_1,~~ \|\tilde{\mathcal{A}}^*(\tilde{b}-\tilde{\mathcal{A}}(X))\|_{S_\infty}\leq\eta_2.
\end{align}

In this paper, we introduce $\ell_q$-RUB condition, which is a natural generalization of RUB condition. And we consider Schatten-$p$ minimization
\begin{align}\label{MatrixSp}
\min_{X}~\|X\|_{S_p}^p ~~\text{subject~ to}~ ~b-\mathcal{A}(X)\in\mathcal{B},
\end{align}
for ROP model (\ref{ROP2}),
where $\mathcal{B}$ is a set determined by the noise structure and $0<p\leq 1$. We consider two types of bounded noises \cite{CW2011}.
One is $l_q$ bounded noises \cite{DET2006}, i.e.,
\begin{align}\label{lpboundednoise}
\mathcal{B}^{l_q}(\eta_1)=\{z: \|z\|_q/L\leq\eta_1\}
\end{align}
for some constant $\eta_1$; and the other is motivated by \textit{Dantzig~Selector} procedure \cite{CT2007}, where
\begin{align}\label{Dantzigselectornoise}
\mathcal{B}^{DS}(\eta_2)=\{z: \|\mathcal{A}^*(z)\|_{S_{\infty}}\leq\eta_2\}
\end{align}
for some constant $\eta_2$. In particular, $\mathcal{B}=\{0\}$ in noiseless case.

And if $\mathcal{A}$ is a symmetric rank-one projection, we use
\begin{align}\label{SROP-MatrixSp}
\min_{X}~\|X\|_{S_p}^p ~~\text{subject~ to}~ ~\tilde{b}-\tilde{\mathcal{A}}(X)\in\mathcal{B}.
\end{align}
instead of (\ref{MatrixSp}).

In the noisy case, we also consider the simpler least $q$-minimization problem
\begin{align}\label{Leastq}
\min_{X\in\mathbb{R}^{m\times n}}~\|\mathcal{A}(X)-b\|_q^q
\end{align}
which may work equally well or even better than Schatten $p$ minimization problem (\ref{MatrixSp}) in terms of recovery under certain natural conditions. Apart from simplicity and computational efficiency (see \cite{LB2010,TW2013} for $q=2$, \cite{LSC2017} for $q=1$), it has the additional advantage that no estimate $\eta$ of the noise level is required. It was proposed by Cand\`{e}s and Li \cite{CL2013} for $q=1$, which is used to solve PhaseLift problem \cite{CESV2013-2015,CSV2013}. They constructed the dual certificate condition to solve this problem.  Later, Kabanava, Kueng, Rauhut et.al. \cite{KKRT2016} considered it for $q\geq1$ for density operators under robust rank null space property.
We should point out that when $q=1$, least $q$-minimization problem in vector case is just the least absolute deviation introduced in \cite{BK1978}. Moreover works about the least absolute deviation, readers can see \cite{P1984,W2012,W2015,WKT2017}.

In this paper, we consider the recovery of the matrix $X\in\mathbb{R}^{m\times n}$ possessing some density, i.e. $\Big(\text{tr}\big((X^{T}X)^{p/2}\big)\Big)^{1/p}=1$, via nonconvex least $q$ minimization. And our method can be written as
\begin{align}\label{DensityLeastq}
\min_{X\in\mathbb{R}^{m\times n}}~\|\mathcal{A}(X)-b\|_q^q~~\text{subject~ to}~~\Big(\text{tr}\big((X^{T}X)^{p/2}\big)\Big)^{1/p}=1,
\end{align}
where $0<p\leq q\leq 1$. And if $\mathcal{A}$ is a symmetric rank-one projection, we use
\begin{align}\label{SROP-DensityLeastq}
\min_{X\in\mathbb{S}^{m}}~\|\tilde{\mathcal{A}}(X)-\tilde{b}\|_q^q~~\text{subject~ to}~~\Big(\text{tr}\big(X^{p}\big)\Big)^{1/p}=1
\end{align}
instead of (\ref{DensityLeastq}).

The contribution of the present work can be summarized as follows.

\begin{itemize}
\item[(1)]We introduce the $\ell_q$-RUB for $0<q\leq 1$, which includes the RUB condition in \cite{CZ2015}.

\item[(2)]A uniform and stable $\ell_q$-RUB condition for low-rank matrices' recovery, via Schatten-$p$ minimization ($0<p<1$), is given for ROP model.  And our condition is also sufficient for SROP model.

\item[(3)]we obtain that the robust rank null sapce property of order $r$ can be deduced from the $\ell_q$ RUB of order $(k+1)r$ for some $k>1$.

\item[(4)]A stable $\ell_q$-RUB condition for low-rank matrices' recovery, via least-$q$ minimization ($0<q\leq1$), is also given for ROP model.

\item[(5)]With high probability, ROP with $L\geq Cr(m+n)$ random projections from Gaussian distribution is sufficient to ensure stable and robust recovery of all rank-$r$ matrices via least absolute deviation estimator.

\end{itemize}

Throughout the article, we use the following basic notations. We denote $\mathbb{Z}_+$ by positive integer set. For any random variable $x$, we use $\mathbb{E}x$ denote the expectation of $x$. For any vector $x\in\mathbb{R}^{L}$ and index set $S\subset \{1,\ldots,L\}$, let $x_S$ be the vector equal to $x$ on $S$ and to zero on $S^c$. And we denote $I_L$ by $L\times L$ identity matrix.
For any matrix $X\in\mathbb{R}^{m\times n}$, we denote $X_{\max(r)}$ as the best rank-$r$ approximation of $X$, and $X_{-\max(r)}=X-X_{\max(r)}$ as the error of the best rank-$r$ approximation of $X$. We use the phrase ``rank-$r$ matrices" to refer to matrices of rank at most $r$.

\section{Recovery via Schatten-$p$ Minimization\label{s2}}
\hskip\parindent

In this section, we introduce $\ell_q$-RUB condition for $0<q\leq 1$ and consider the recovery of matrices through Schatten-$p$ minimization (\ref{MatrixSp}). We show that the $\ell_q$-RUB condition of order $(k+1)r$ for any $k>1$ such that $kr\in\mathbb{Z}_+$, with constants $C_1,C_2$ satisfies $C_2/{C_1}<k^{(1/p-1/2)q}$ for $0<p\leq q\leq 1$ is sufficient to guarantee the exact and and Stable recovery of all rank-$r$ matrices.

\subsection{$\ell_q$-RUB and Some Auxiliary Lemmas \label{s2.1}}
\quad

In this subsection, we will introduce $\ell_q$-RUB condition for $0<q\leq 1$ and give some auxiliary lemmas.
Firstly, we introduce $\ell_q$-RUB condition, which is a natural generalization of RUB in \cite{CZ2015}.

\begin{definition}($\ell_q$-Restricted Uniform Boundedness)\label{lqRUB}
For a linear map $\mathcal{A}:\mathbb{R}^{m\times n}\rightarrow \mathbb{R}^{L}$, a positive integer $r$ and $0<q\leq 1$, if there exist uniform constants $C_1$ and $C_2$ such that for all nonzero rank-$r$ matrices $X\in\mathbb{R}^{m\times n}$
\begin{align}\label{e2.1}
C_1\|X\|_{S_2}^q\leq\|\mathcal{A}(X)\|_q^q/L\leq C_2\|X\|_{S_2}^q,
\end{align}
we say that linear map $\mathcal{A}$ satisfies $\ell_q$-Restricted Uniform Boundedness ($\ell_q$-RUB) condition of order $r$ with constants $C_1$ and $C_2$.
\end{definition}

\begin{remark}(RIP-$\ell_2/\ell_q$ for low-rank matrices)\label{MatrixRIPl2lq}
If we take $C_1=1-\delta_r^{\text{lb}}$ and $C_2=1+\delta_r^{\text{sub}}$, then (\ref{e2.1}) becomes
\begin{align*}
(1-\delta_r^{\text{lb}})\|X\|_{S_2}^q\leq\frac{1}{L}\|\mathcal{A}(X)\|_q^q\leq(1+\delta_r^{\text{sub}})\|X\|_{S_2}^q.
\end{align*}
And we call that the linear map $\mathcal{A}$ satisfies RIP-$\ell_2/{\ell_q}$ of order $r$ with constants $\delta_r^{\text{lb}}$ and $\delta_r^{\text{sub}}$. Especially, when $q=1$, it is the RIP-$\ell_2/{\ell_1}$ for low-rank matrices introduced in \cite{CCG2015}. When $\delta_r^{\text{lb}}=\delta_r^{\text{sub}}$, it is the restricted $q$-isometry property introduced in \cite{ZHZ2013}.
\end{remark}

Then we will give some auxiliary lemmas.

The first lemma give out the condition guaranteeing the additivity of Schatten-$p$ norm, which comes from \cite[Lemma 2.3]{RFP2010} for $q=1$, and\cite[Lemma
2.2]{KX2013} and \cite[Lemma 2.1]{ZHZ2013} $0<q<1$.

\begin{lemma}\label{orthogonaldecomposition}
Let $0<q\leq 1$. Let $X,Y\in\mathbb{R}^{m\times n}$ be matrices with $X^{T}Y=O$ and $XY^{T}=O$,  then the following holds:
\begin{itemize}
\item[(1)]$\|X+Y\|_{S_q}^q=\|X\|_{S_q}^q+\|Y\|_{S_q}^q$;
\item[(2)]$\|X+Y\|_{S_q}\geq\|X\|_{S_q}+\|Y\|_{S_q}$.
\end{itemize}
\end{lemma}

And the second one, which comes from \cite[Lemma 2.6]{CL2015}, is the cone constraint for matrix's Schatten-$p$ norm.
\begin{lemma}\label{MatrixConeconstraintinequality}
Suppose $X,\hat{X}\in\mathbb{R}^{m\times n}$, $R=\hat{X}-X$. If $\|\hat{X}\|_{S_p}^p\leq \|X\|_{S_p}^p$, then we have
$$
\|R_{-max(r)}\|_{S_p}^p\leq2\|X_{-max(r)}\|_{S_p}^p+\|R_{max(r)}\|_{S_p}^p.
$$
\end{lemma}

In order to estimate the Schatten-2 norm of $(\hat{X}-X)_{-\max(r)}$ for the stable recovery, we also need the following lemma, which is  inspired by \cite[Lemma 7.8]{CZ2015}.

\begin{lemma}\label{Vitaltechnique}
Let
$$
f(t)=t^{(2-p)/p}\big(c-kt\big),
$$
where $c$ is a positive constant independent of $t$. Then we have
\begin{align*}
f(t)\leq\frac{p}{2}\bigg(\frac{2-p}{2k}\bigg)^{(2-p)/p}c^{2/p},~t>0.
\end{align*}
\end{lemma}
\begin{proof}
By taking derivation of $f(t)$, we have
$$
f'(t)=-\frac{1}{p}t^{(2-2p)/p}\big(2kt-(2-p)c\big).
$$
Therefore
\begin{align*}
f(t)&\leq f\bigg(\frac{(2-p)}{2k}c\bigg)=\frac{p}{2}\bigg(\frac{2-p}{2k}\bigg)^{(2-p)/p}c^{2/p}.
\end{align*}
\end{proof}

\subsection{Exact Recovery via via Schatten-$p$ Minimization \label{s2.2}}
\quad

In this subsection, we will consider the exact recovery under $\ell_q$-RUB condition.

\begin{theorem}\label{ROP-lp-Exact}
Let $r$ be any positive integer, and $k>1$ such that $kr$ be positive integer. Suppose that $\mathcal{A}$ satisfies $\ell_q$-RUB of order $(k+1)r$ with $C_2/{C_1}<k^{(1/p-1/2)q}$ for any $0<p\leq q\leq 1$, then the Schatten-$p$ minimization (\ref{MatrixSp}) exact recovers all rank-$r$ matrices. That is, for all rank-$r$ matrices $X$ with $b=\mathcal{A}(X)$, we have $\hat{X}=X$, where $\hat{X}$ is the solution of (\ref{MatrixSp}) with $\mathcal{B}=\{0\}$.
\end{theorem}

Before proving Theorems \ref{ROP-lp-Exact}, let us first state the well known matrix null space property, which lies in the heart of the proof of this main result.
\begin{lemma}(\cite{YS2016})\label{Exactrecovery-NSP}
All rank-$r$ matrices $X\in\mathbb{R}^{m\times n}$ with $b=\mathcal{A}(X)\in\mathbb{R}^{L}$ can be exactly recovered by solving problem (\ref{MatrixSp}) with $\mathcal{B}=\{0\}$ if and only if
$$
\|R_{\max(r)}\|_{S_p}^p<\|R_{-max(r)}\|_{S_p}^p
$$
holds for all $R\in\mathcal{N}(\mathcal{A})\backslash\{O\}$, where $O$ is $m\times n$ zero matrix.
\end{lemma}

\begin{proof}[Proof of Theorem \ref{ROP-lp-Exact}]
Our proof follows the idea of \cite[Theorem 2.4]{CS2008}.
By Lemma \ref{Exactrecovery-NSP}, we only need to show that for all matrices $R\in\mathcal{N}(\mathcal{A})\backslash \{O\}$, one has
$\|R_{\max(r)}\|_{S_p}^p<\|R_{-max(r)}\|_{S_p}^p$. Assume there exists a nonzero matrix $R$ with
\begin{align}\label{e2.2}
\mathcal{A}(R)=0
\end{align}
and
\begin{align}\label{e2.3}
\|R_{-\max(r)}\|_{S_p}^p\leq\|R_{max(r)}\|_{S_p}^p.
\end{align}
We denote $l=\min\{m,n\}$ and assume that the singular value decomposition of $R$ is
$$
R=U\Sigma V^{T}=U\text{diag}(\sigma)V^{T},
$$
where $\sigma=(\sigma_1,\ldots,\sigma_l)$ is the singular vector with $\sigma_1\geq\sigma_2\geq\cdots\geq\sigma_l\geq 0$. Without loss of generality, we assume that $(k+1)r\leq l$, otherwise we set the undefined entries of $\sigma$ as 0.

Let $\text{supp}(\sigma_{\max(r)})=\{1,\ldots,r\}=:T_0$. We partition $T_0^c=\{r+1,\ldots,l\}$ as
$$
T_0^c=\cup_{j=1}^J T_j,
$$
where $T_1$ is the index set of the $s$ largest entries of $\sigma_{-\max(r)}$, $T_2$ is the index set of the next $s$ largest entries of $\sigma_{-\max(r)}$, and so on. Here, $s\in\mathbb{Z}_+$ is to be determined. The last index set $T_J$ may contain less $s$ elements.
Then for $j\geq 2$, we have $|\sigma_i|^p\leq\|\sigma_{T_{j-1}}\|_p^p/s$ for any $i\in T_{j}$. Therefore
\begin{align*}
\|\sigma_{T_j}\|_2^2=\sum_{i\in T_j}|\sigma_i|^2\leq\sum_{i\in T_j}\bigg(\frac{\|\sigma_{T_{j-1}}\|_p^p}{s}\bigg)^{2/p}
\leq \frac{\|\sigma_{T_{j-1}}\|_p^2}{s^{2/p-1}},
\end{align*}
so that
\begin{align}\label{e2.4}
\sum_{j\geq 2}\|\sigma_{T_j}\|_2^q&=\sum_{j\geq 2}\big(\|\sigma_{T_j}\|_2^2\big)^{q/2}
\leq\sum_{j\geq 2}\bigg(\frac{\|\sigma_{T_{j-1}}\|_p^2}{s^{2/p-1}}\bigg)^{q/2}
=\frac{1}{\big(s^{1/p-1/2}\big)^q}\sum_{j\geq 2}\big(\|\sigma_{T_{j-1}}\|_p^p\big)^{q/p}\nonumber\\
&\leq\frac{1}{\big(s^{1/p-1/2}\big)^q}\bigg(\sum_{j\geq 2}\|\sigma_{T_{j-1}}\|_p^p\bigg)^{q/p}
=\frac{1}{\big(s^{1/p-1/2}\big)^q}\big(\|\sigma_{-\max(r)}\|_p^p\big)^{q/p},
\end{align}
where the second line follows from $q/p\geq 1$.

Let $T_{01}=T_0\cup T_1$, we consider the following identity
\begin{align}\label{e2.5}
\|\mathcal{A}(R)\|_q^q=\bigg\|\mathcal{A}(R_{T_{01}})+\sum_{j\geq 2}\mathcal{A}(R_{T_{j}})\bigg\|_q^q,
\end{align}
where
$$
R_{T_j}=U\text{diag}(\sigma_{T_j}) V^{T},~\forall j=0,1,\ldots.
$$

First, we give out a lower bound for (\ref{e2.5}). Note that $\text{rank}(R_{T_{01}})\leq r+s$ and $|T_j|\leq s$. By the $r+s$ order $\ell_q$-RUB condition, we have
\begin{align}\label{e2.6}
\|\mathcal{A}(R)\|_q^q&\geq\|\mathcal{A}(R_{T_{01}})\|_q^q-\sum_{j\geq 2}\|\mathcal{A}(R_{T_{j}})\|_q^q\nonumber\\
&\geq C_1L\|R_{T_{01}}\|_{S_2}^q-\sum_{j\geq 2}C_2L\|R_{T_j}\|_{S_2}^q
=C_1L\|R_{T_{01}}\|_{S_2}^q-C_2L\sum_{j\geq 2}\|\sigma_{T_j}\|_2^q\nonumber\\
&\geq C_1L\|R_{T_{01}}\|_{S_2}^q-\frac{C_2L}{\big(s^{1/p-1/2}\big)^q}\big(\|\sigma_{-\max(r)}\|_p^p\big)^{q/p},
\end{align}
where the last inequality follows from (\ref{e2.4}). Then by (\ref{e2.3}), we get a lower bound of $\|\mathcal{A}(R)\|_q^q$ as follows
\begin{align}\label{e2.7}
\|\mathcal{A}(R)\|_q^q&\geq C_1L\|R_{T_{01}}\|_{S_2}^q
-\frac{C_2L}{\big(s^{1/p-1/2}\big)^q}\big(\|\sigma_{\max(r)}\|_p^p\big)^{q/p}\nonumber\\
&\geq C_1L\|R_{T_{01}}\|_{S_2}^q
-C_2L\bigg(\frac{r}{s}\bigg)^{(1/p-1/2)q}\|\sigma_{\max(r)}\|_2^q\nonumber\\
&=C_1L\|R_{T_{01}}\|_{S_2}^q
-C_2L\bigg(\frac{r}{s}\bigg)^{(1/p-1/2)q}\|R_{\max(r)}\|_{S_2}^q\nonumber\\
&\geq L\bigg(C_1-C_2\bigg(\frac{r}{s}\bigg)^{(1/p-1/2)q}\bigg)\|R_{T_{01}}\|_{S_2}^q.
\end{align}

We also need an upper bound of $\|\mathcal{A}(R)\|_q^q$. Using (\ref{e2.2}), we have
\begin{align}\label{e2.8}
\|\mathcal{A}(R)\|_q^q=0.
\end{align}

Combining the lower bound (\ref{e2.7}) with the upper bound (\ref{e2.8}), we get
\begin{align}\label{e2.9}
L\bigg(C_1-C_2\bigg(\frac{r}{s}\bigg)^{(1/p-1/2)q}\bigg)\|R_{T_{01}}\|_{S_2}^q\leq0.
\end{align}
Taking $s=kr$ for $k>1$ with $kr\in\mathbb{Z}_+$. Note that
$$
\frac{C_2}{C_1}< \bigg(\frac{s}{r}\bigg)^{(1/p-1/2)q}=k^{(1/p-1/2)q}.
$$
Then (\ref{e2.9}) implies that
\begin{align*}
\|R_{\max(r)}\|_{S_p}\leq r^{1/p-1/2}\|R_{\max(r)}\|_{S_2}\leq r^{1/p-1/2}\|R_{T_{01}}\|_{S_2}\leq 0.
\end{align*}
However, (\ref{e2.2}) implies that
\begin{align}
0\leq\|R_{-\max(r)}\|_{S_p}\leq\|R_{max(r)}\|_{S_p},
\end{align}
which is a contradiction.
\end{proof}

As a direct consequence of Theorem \ref{ROP-lp-Exact}, we have following two corollaries.
\begin{corollary}\label{SROP-lp-Exact}
Let $r$ be any positive integer, and $k>1$ such that $kr$ be positive integer. Suppose that $\tilde{\mathcal{A}}$ satisfies $\ell_q$ RUB of order $(k+1)r$ with $C_2/{C_1}<k^{(1/p-1/2)q}$ for any $0<p\leq q\leq 1$, then the Schatten-$p$ minimization (\ref{SROP-MatrixSp}) exact recovers all symmetric rank-$r$ matrices. That is, for all symmetric rank-$r$ matrices $X$ with $\tilde{b}=\tilde{\mathcal{A}}(X)$, we have $\hat{X}=X$, where $\hat{X}$ is the solution of (\ref{SROP-MatrixSp}) with $\mathcal{B}=\{0\}$.
\end{corollary}

\begin{corollary}\label{Corollary-ROP-lp-Exact}
Let $\tau>1$ and $s=\lceil r\tau^{\frac{2p}{(2-p)q}}\rceil$.
Suppose that $\mathcal{A}$ satisfies RIP-$\ell_2/\ell_q$ of order $s+r$ with
$$
\delta_{s+r}^{\text{sub}}+\tau\delta_{r}^{\text{lb}}<\tau-1
$$
for any $0<p\leq q\leq 1$, then the Schatten-$p$ minimization (\ref{MatrixSp}) recovers all rank-$r$ matrices. That is, for all rank-$r$ matrices $X$ with $b=\mathcal{A}(X)$, we have $\hat{X}=X$, where $\hat{X}$ is the solution of (\ref{MatrixSp}) with $\mathcal{B}=\{0\}$.
\end{corollary}

However, our results may not be optimal and can be improved furthermore.
\begin{remark}\label{Openproblem1}
We emphasize that our $(k+1)r$ order $\ell_q$-RUB condition for $q=1$ is a litter stronger than the condition $kr$ order RUB condition in \cite{CZ2015}. We note that in \cite{CZ2015}, Cai and Zhang used a sparse representation of a polytope in $\ell_1$ norm (see  \cite[Lemma 1.1]{CZ2014}), which provide a more refined analysis. And Zhang and Li \cite[Lemma 2.1]{ZL2017-1} gave out a similar sparse representation in $\ell_p$ norm for $0<p\leq 1$. However this sparse representation can not be direct used in our model (\ref{MatrixSp}). Therefore, we don't know whether
or not it is possible to reduce this $(k+1)r$ order $\ell_q$-RUB condition to $kr$ order for $k>1$ with $kr\in\mathbb{Z}_+$.
\end{remark}

\subsection{Stable Recovery via Schatten-$p$ Minimization \label{s2.3}}
\quad

In this subsection, we consider the stable recovery. We show that the $\ell_q$-RUB condition of order $(k+1)r$  with constants $C_1,C_2$ satisfies $C_2/{C_1}<k^{(1/p-1/2)q}$ is also sufficient to guarantee the stable recovery of all rank-$r$ matrices via the noisy measurements.

\begin{theorem}\label{ROP-lp-lq-DS}
Let $r$ be any positive integer and $0<p\leq q\leq 1$.  Let $\hat{X}^{\ell_q}$ be the solution of the Schatten-$p$ minimization (\ref{MatrixSp}) with $\mathcal{B}=\mathcal{B}^{\ell_q}(\eta_1)\cap \mathcal{B}^{DS}(\eta_2)$.
\begin{itemize}
\item[(1)]
For any rank-$r$ $X$, let $k>1$ with $kr$ be positive integer, and $\tilde{\mathcal{A}}$ satisfies $\ell_q$-RUB of order $(k+1)r$ with $C_2/{C_1}<k^{(1/p-/2)q}$, then
\begin{align*}
\|\hat{X}&-X\|_{S_2}^q\leq
\Bigg(\bigg(\frac{1}{k}\bigg)^{(1/p-1/2)q}+1\Bigg)\frac{1}{\rho_1 L^{q}}
\min\bigg\{\frac{2}{L^{1-2q}}\eta_1^q,\frac{2^{q/p+q}}{\rho_1}r^{(1/p-1/2)q}\eta_2^q\bigg\},
\end{align*}
where $\rho_1=C_1-C_2\big(1/k\big)^{(1/p-1/2)q}$.
\item[(2)]
And for more general matrix $X$, let $k>2^{\frac{2(q-p)}{q(2-p)}}$ with $kr$ be positive integer, and $\tilde{\mathcal{A}}$ satisfies $\ell_q$-RUB of order $(k+1)r$ with $C_2/{C_1}<2^{1-q/p}k^{(1/p-/2)q}$, then
\begin{align*}
\|\hat{X}&-X\|_{S_2}^q\\
&\leq \Bigg(\frac{C_22^{2q/p-1}}{\rho_2}\bigg(\frac{1}{k}\bigg)^{(1/p-1/2)q}+1\Bigg)
\Bigg(2^{2q/p-1}\bigg(\frac{1}{k}\bigg)^{(1/p-1/2)q}+1\Bigg)\bigg(\frac{\|X_{-\max(r)}\|_{S_p}}{r^{1/p-1/2}}\bigg)^q\nonumber\\
&\hspace*{12pt}+\Bigg(2^{q/p-1}\bigg(\frac{1}{k}\bigg)^{(1/p-1/2)q}+1\Bigg)\frac{1}{\rho_2 L^{q}}
\min\bigg\{\frac{2}{L^{1-2q}}\eta_1^q,\frac{2^{2q/p+q+1}}{\rho_1}r^{(1/p-1/2)q}\eta_2^q\bigg\},
\end{align*}
where $\rho_2=C_1-C_22^{q/p-1}\big(1/k\big)^{(1/p-1/2)q}$.
\end{itemize}
\end{theorem}

In fact, Theorem \ref{ROP-lp-lq-DS} can be deduced by following Proposition \ref{ROP-lp-lq} for $\mathcal{B}^{\ell_q}(\eta_1)$ and Proposition \ref{ROP-lp-DS} for $\mathcal{B}^{DS}(\eta_2)$.
First, we give out the estimate $\|\hat{X}-X\|_{S_2}^q$ for the $\mathcal{B}=\mathcal{B}^{\ell_q}(\eta_1)$.

\begin{proposition}\label{ROP-lp-lq}
Let $r$ be any positive integer and $0<p\leq q\leq 1$.  Let $\hat{X}^{\ell_q}$ be the solution of the Schatten-$p$ minimization (\ref{MatrixSp}) with $\mathcal{B}=\mathcal{B}^{\ell_q}(\eta_1)$.
\begin{itemize}
\item[(1)]
For any rank-$r$ $X$, let $k>1$ with $kr$ be positive integer, and $\mathcal{A}$ satisfies $\ell_q$-RUB of order $(k+1)r$ with $C_2/{C_1}<k^{(1/p-/2)q}$, then
\begin{align*}
\|\hat{X}^{\ell_q}&-X\|_{S_2}^q\leq\frac{2}{\rho_1 L^{1-q}}
\Bigg(\bigg(\frac{1}{k}\bigg)^{(1/p-1/2)q}+1\Bigg)\eta_1^q,
\end{align*}
where $\rho_1=C_1-C_2\big(1/k\big)^{(1/p-1/2)q}$.
\item[(2)]
And for more general matrix $X$, let $k>2^{\frac{2(q-p)}{q(2-p)}}$ with $kr$ be positive integer, and $\mathcal{A}$ satisfies $\ell_q$-RUB of order $(k+1)r$ with $C_2/{C_1}<2^{1-q/p}k^{(1/p-/2)q}$, then
\begin{align*}
\|\hat{X}^{\ell_q}&-X\|_{S_2}^q\\
&\leq \Bigg(\frac{C_22^{2q/p-1}}{\rho_2}\bigg(\frac{1}{k}\bigg)^{(1/p-1/2)q}+1\Bigg)\Bigg(2^{2q/p-1}
\bigg(\frac{1}{k}\bigg)^{(1/p-1/2)q}+1\Bigg)\bigg(\frac{\|X_{-\max(r)}\|_{S_p}}{r^{1/p-1/2}}\bigg)^q\nonumber\\
&\hspace*{12pt}+\frac{2}{\rho_2 L^{1-q}}\Bigg(2^{q/p-1}
\bigg(\frac{1}{k}\bigg)^{(1/p-1/2)q}+1\Bigg)\eta_1^q,
\end{align*}
where $\rho_2=C_1-C_22^{q/p-1}\big(1/k\big)^{(1/p-1/2)q}$.
\end{itemize}
\end{proposition}

\begin{proof}
We denote $l=\min\{m,n\}$. Let $R=\hat{X}^{\ell_q}-X$. We have the following tube constraint inequality
\begin{align}\label{e3.1}
\|\mathcal{A}(R)\|_q^q\leq\|\mathcal{A}(\hat{X}^{\ell_p})-b\|_q^q+\|b-\mathcal{A}(X)\|_q^q
\leq(L\eta_1)^q+(L\eta_1)^q=2L^{q}\eta_1^q.
\end{align}
By Lemma \ref{MatrixConeconstraintinequality}, we have cone constraint inequality as follows
\begin{align}\label{e3.2}
\|R_{-max(r)}\|_{S_p}^p\leq2\|X_{-max(r)}\|_{S_p}^p+\|R_{max(r)}\|_{S_p}^p.
\end{align}

We still assume that the singular value decomposition of $R$ is
$$
R=U\Sigma V^{T}=U\text{diag}(\sigma)V^{T},
$$
where $\sigma=(\sigma_1,\ldots,\sigma_l)$ is the singular vector with $\sigma_1\geq\sigma_2\geq\cdots\geq\sigma_l\geq 0$. And we denote $T_0:=\text{supp}(\sigma_{\max(r)})=\{1,\ldots,r\}$. We also partition $T_0^c=\{r+1,\ldots,l\}$ as
$$
T_0^c=\cup_{j=1}^J T_j,
$$
which is the same as the proof of Theorem \ref{ROP-lp-Exact}. And (\ref{e2.4}) still holds.

Let $T_{01}=T_0\cup T_1$, we also consider identity (\ref{e2.5}).

First, by (\ref{e2.6}) and (\ref{e3.2}),  we can estimate a lower bound for (\ref{e2.5}) as follows
\begin{align*}
\|\mathcal{A}(R)\|_q^q&\geq C_1L\|R_{T_{01}}\|_{S_2}^q
-\frac{C_2L}{\big(s^{1/p-1/2}\big)^q}\big(\|\sigma_{\max(r)}\|_p^p+2\|X_{-\max(r)}\|_{S_p}^p\big)^{q/p}\\
&\geq C_1L\|R_{T_{01}}\|_{S_2}^q
-C_2L\bigg(\frac{r}{s}\bigg)^{(1/p-1/2)q}
\Bigg(\|\sigma_{\max(r)}\|_2^p+2\bigg(\frac{\|X_{-\max(r)}\|_{S_p}}{r^{1/p-1/2}}\bigg)^p\Bigg)^{q/p}\\
\end{align*}
It follows from $(|a|+|b|)^{q/p}\leq 2^{q/p-1}\big(|a|^{q/p}+|b|^{q/p}\big)$ for $q/p\geq 1$ that
\begin{align*}
\Bigg(\|\sigma_{\max(r)}\|_2^p+2\bigg(\frac{\|X_{-\max(r)}\|_{S_p}^p}{r^{1/p-1/2}}\bigg)^p\Bigg)^{q/p}
\leq 2^{q/p-1}\Bigg(\|\sigma_{\max(r)}\|_2^q+2^{q/p}\bigg(\frac{\|X_{-\max(r)}\|_{S_p}}{r^{1/p-1/2}}\bigg)^q\Bigg).
\end{align*}
But when $X$ is rank-$r$, i.e., $\|X_{-\max(r)}\|_{S_p}^p=0$, we have a simpler form
\begin{align*}
\Bigg(\|\sigma_{\max(r)}\|_2^p+2\bigg(\frac{\|X_{-\max(r)}\|_{S_p}^p}{r^{1/p-1/2}}\bigg)^p\Bigg)^{q/p}
=\|\sigma_{\max(r)}\|_2^q.
\end{align*}
Therefore
\begin{align}\label{e3.3}
\|\mathcal{A}(R)\|_q^q
&\geq C_1L\|R_{T_{01}}\|_{S_2}^q-C_2L2^{q/p-1}\bigg(\frac{r}{s}\bigg)^{(1/p-1/2)q}
\Bigg(\|\sigma_{\max(r)}\|_2^q+2^{q/p}\bigg(\frac{\|X_{-\max(r)}\|_{S_p}}{r^{1/p-1/2}}\bigg)^q\Bigg)\nonumber\\
&\geq L\bigg(C_1-C_22^{q/p-1}\bigg(\frac{r}{s}\bigg)^{(1/p-1/2)q}\bigg)\|R_{T_{01}}\|_{S_2}^q
-C_2L2^{2q/p-1}\bigg(\frac{r}{s}\bigg)^{(1/p-1/2)q}\bigg(\frac{\|X_{-\max(r)}\|_{S_p}}{r^{1/p-1/2}}\bigg)^q
\end{align}
or
\begin{align}\label{e3.3-1}
\|\mathcal{A}(R)\|_q^q
&\geq C_1L\|R_{T_{01}}\|_{S_2}^q-C_2L\bigg(\frac{r}{s}\bigg)^{(1/p-1/2)q}
\|\sigma_{\max(r)}\|_2^q\nonumber\\
&\geq L\bigg(C_1-C_2\bigg(\frac{r}{s}\bigg)^{(1/p-1/2)q}\bigg)\|R_{T_{01}}\|_{S_2}^q.
\end{align}

We also need an upper bound of $\|\mathcal{A}(R)\|_q^q$. Using (\ref{e3.1}), we have
\begin{align}\label{e3.4}
\|\mathcal{A}(R)\|_q^q\leq 2L^q\eta_1^q.
\end{align}

Combining the lower bound (\ref{e3.3}) (or (\ref{e3.3-1})) with the upper bound (\ref{e3.4}), we get
\begin{align}\label{e3.5}
L\bigg(C_1-C_22^{q/p-1}\bigg(\frac{r}{s}\bigg)^{(1/p-1/2)q}\bigg)\|R_{T_{01}}\|_{S_2}^q
-C_2L2^{2q/p-1}\bigg(\frac{r}{s}\bigg)^{(1/p-1/2)q}\bigg(\frac{\|X_{-\max(r)}\|_{S_p}}{r^{1/p-1/2}}\bigg)^q
\leq  2L^q\eta_1^q
\end{align}
or
\begin{align}\label{e3.5-1}
L\bigg(C_1-C_2\bigg(\frac{r}{s}\bigg)^{(1/p-1/2)q}\bigg)\|R_{T_{01}}\|_{S_2}^q
\leq  2L^q\eta_1^q
\end{align}

Taking $s=kr$ with $kr\in\mathbb{Z}_+$. Note that
\begin{align*}
\rho=
\begin{cases}
C_1-C_2\big(\frac{1}{k}\big)^{(1/p-1/2)q}>0, &\text{if}~X \text{~is~rank}~r\\
C_1-C_22^{q/p-1}\big(\frac{1}{k}\big)^{(1/p-1/2)q}>0, &\text{otherwise}.
\end{cases}
\end{align*}
Therefore
\begin{align}\label{e3.6}
\|R_{T_{01}}\|_{S_2}^q\leq \frac{C_22^{2q/p-1}}{\rho}\bigg(\frac{1}{k}\bigg)^{(1/p-1/2)q}\bigg(\frac{\|X_{-\max(r)}\|_{S_p}}{r^{1/p-1/2}}\bigg)^q
+\frac{2}{\rho L^{1-q}}\eta_1^q
\end{align}
or
\begin{align}\label{e3.6-1}
\|R_{T_{01}}\|_{S_2}^q\leq\frac{2}{\rho L^{1-q}}\eta_1^q.
\end{align}

Next, we estimate $\|R_{T_{01}^c}\|_{S_2}^q$. Here, we use some idea of the proof of \cite[Lemma 7.8]{CZ2015}. First, by a simply computation, we have
\begin{align}\label{e3.7}
\|R_{T_{01}^c}\|_{S_2}^q&=\|\sigma_{-\max{(s+r)}}\|_2^q
\leq\Big(\big(\|\sigma_{-\max(s+r)}\|_{\infty}^{p}\big)^{(2-p)/p}\|\sigma_{-\max(s+r)}\|_p^p\Big)^{q/2}\nonumber\\
&\leq\Bigg((\sigma_{s+r}^p)^{(2-p)/p}
\bigg(\|\sigma_{-\max(r)}\|_p^p-\sum_{j=r+1}^{s+r}\sigma_j^p\bigg)\Bigg)^{q/2}\nonumber\\
&\leq\Bigg((\sigma_{s+r}^p)^{(2-p)/p}
\bigg(\big(\|\sigma_{\max(r)}\|_p^p+2\|X_{-\max(r)}\|_{S_p}^p\big)-s\sigma_{s+r}^p\bigg)\Bigg)^{q/2}\nonumber\\
&:=\big(f(\sigma_{s+r}^p)\big)^{q/2},
\end{align}
where
$$
f(t)=t^{(2-p)/p}\Big(\big(\|\sigma_{\max(r)}\|_p^p+2\|X_{-\max(r)}\|_{S_p}^p\big)-st\Big).
$$
We need to estimate the upper bound of $f(\sigma_{s+r}^p)$. By Lemma \ref{Vitaltechnique}, we have
\begin{align*}
f(t)&\leq\frac{p}{2}\bigg(\frac{2-p}{2s}\bigg)^{(2-p)/p}\big(\|\sigma_{\max(r)}\|_p^p+2\|X_{-\max(r)}\|_{S_p}^p\big)^{2/p}\\
&\leq \frac{p}{2}\bigg(\frac{2-p}{2}\bigg)^{(2-p)/p}\bigg(\frac{1}{k}\bigg)^{(1/p-1/2)2}
\Bigg(\|\sigma_{\max(r)}\|_2^p+2\bigg(\frac{\|X_{-\max(r)}\|_{S_p}}{r^{1/p-1/2}}\bigg)^p\Bigg)^{2/p},
\end{align*}
which implies that
\begin{align}\label{e3.8}
\|R_{T_{01}^c}\|_{S_2}^q&\leq \bigg(\frac{p}{2}\bigg)^{q/2}
\bigg(\frac{2-p}{2}\bigg)^{(1/p-1/2)q}\bigg(\frac{1}{k}\bigg)^{(1/p-1/2)q}
\Bigg(\|\sigma_{\max(r)}\|_2^p+2\bigg(\frac{\|X_{-\max(r)}\|_{S_p}}{r^{1/p-1/2}}\bigg)^p\Bigg)^{q/p}\nonumber\\
&\leq2^{q/p-1}\bigg(\frac{p}{2}\bigg)^{q/2}
\bigg(\frac{2-p}{2}\bigg)^{(1/p-1/2)q}\bigg(\frac{1}{k}\bigg)^{(1/p-1/2)q}
\Bigg(\|\sigma_{\max(r)}\|_2^q+2^{q/p}\bigg(\frac{\|X_{-\max(r)}\|_{S_p}}{r^{1/p-1/2}}\bigg)^q\Bigg)\nonumber\\
&\leq2^{q/p-1}\bigg(\frac{p}{2}\bigg)^{q/2}
\bigg(\frac{2-p}{2}\bigg)^{(1/p-1/2)q}\bigg(\frac{1}{k}\bigg)^{(1/p-1/2)q}
\Bigg(\|R_{T_{01}}\|_{S_2}^q+2^{q/p}\bigg(\frac{\|X_{-\max(r)}\|_{S_p}}{r^{1/p-1/2}}\bigg)^q\Bigg).
\end{align}
or
\begin{align}\label{e3.8-1}
\|R_{T_{01}^c}\|_{S_2}^q&\leq \bigg(\frac{p}{2}\bigg)^{q/2}
\bigg(\frac{2-p}{2}\bigg)^{(1/p-1/2)q}\bigg(\frac{1}{k}\bigg)^{(1/p-1/2)q}
\big(\|\sigma_{\max(r)}\|_2^p\big)^{q/p}\nonumber\\
&\leq\bigg(\frac{p}{2}\bigg)^{q/2}
\bigg(\frac{2-p}{2}\bigg)^{(1/p-1/2)q}\bigg(\frac{1}{k}\bigg)^{(1/p-1/2)q}\|R_{T_{01}}\|_{S_2}^q.
\end{align}

It follows from (\ref{e3.8}) (or \ref{e3.8-1}) that
\begin{align}\label{e3.9}
\|R\|_{S_2}^q&\leq \|R_{T_{01}}\|_{S_2}^q+\|R_{T_{01}^c}\|_{S_2}^q\nonumber\\
&\leq \Bigg(2^{q/p-1}\bigg(\frac{p}{2}\bigg)^{q/2}
\bigg(\frac{2-p}{2}\bigg)^{(1/p-1/2)q}\bigg(\frac{1}{k}\bigg)^{(1/p-1/2)q}+1\Bigg)\|R_{T_{01}}\|_{S_2}^q\nonumber\\
&\hspace*{12pt}+2^{2q/p-1}\bigg(\frac{p}{2}\bigg)^{q/2}
\bigg(\frac{2-p}{2}\bigg)^{(1/p-1/2)q}\bigg(\frac{1}{k}\bigg)^{(1/p-1/2)q}
\bigg(\frac{\|X_{-\max(r)}\|_{S_p}}{r^{1/p-1/2}}\bigg)^q
\end{align}
or
\begin{align}\label{e3.9-1}
\|R\|_{S_2}^q&\leq \|R_{T_{01}}\|_{S_2}^q+\|R_{T_{01}^c}\|_{S_2}^q\nonumber\\
&\leq \Bigg(\bigg(\frac{p}{2}\bigg)^{q/2}
\bigg(\frac{2-p}{2}\bigg)^{(1/p-1/2)q}\bigg(\frac{1}{k}\bigg)^{(1/p-1/2)q}+1\Bigg)\|R_{T_{01}}\|_{S_2}^q.
\end{align}

Then substituting (\ref{e3.6}) into (\ref{e3.9}), we have
\begin{align}\label{e3.10}
\|R\|_{S_2}^q&\leq
\Bigg\{\frac{C_22^{2q/p-1}}{\rho}\bigg(\frac{1}{k}\bigg)^{(1/p-1/2)q}\Bigg(2^{q/p-1}\bigg(\frac{p}{2}\bigg)^{q/2}
\bigg(\frac{2-p}{2}\bigg)^{(1/p-1/2)q}\bigg(\frac{1}{k}\bigg)^{(1/p-1/2)q}+1\Bigg)\nonumber\\
&\hspace*{12pt}+2^{2q/p-1}\bigg(\frac{p}{2}\bigg)^{q/2}
\bigg(\frac{2-p}{2}\bigg)^{(1/p-1/2)q}\bigg(\frac{1}{k}\bigg)^{(1/p-1/2)q}\Bigg\}
\bigg(\frac{\|X_{-\max(r)}\|_{S_p}}{r^{1/p-1/2}}\bigg)^q\nonumber\\
&\hspace*{12pt}+\frac{2}{\rho L^{1-q}}\Bigg(2^{q/p-1}\bigg(\frac{p}{2}\bigg)^{q/2}
\bigg(\frac{2-p}{2}\bigg)^{(1/p-1/2)q}\bigg(\frac{1}{k}\bigg)^{(1/p-1/2)q}+1\Bigg)\eta_1^q\nonumber\\
&\leq \Bigg(\frac{C_22^{2q/p-1}}{\rho}\bigg(\frac{1}{k}\bigg)^{(1/p-1/2)q}+1\Bigg)\Bigg(2^{2q/p-1}
\bigg(\frac{1}{k}\bigg)^{(1/p-1/2)q}+1\Bigg)\bigg(\frac{\|X_{-\max(r)}\|_{S_p}}{r^{1/p-1/2}}\bigg)^q\nonumber\\
&\hspace*{12pt}+\frac{2}{\rho L^{1-q}}\Bigg(2^{q/p-1}
\bigg(\frac{1}{k}\bigg)^{(1/p-1/2)q}+1\Bigg)\eta_1^q,
\end{align}
where the last inequality follows from $0<p/2\leq 1$ and $(1/p-1/2)q>0$.

And substituting (\ref{e3.6-1}) into (\ref{e3.9-1}), we have
\begin{align}\label{e3.10-1}
\|R\|_{S_2}^q&\leq
\frac{2}{\rho L^{1-q}}\Bigg(\bigg(\frac{1}{k}\bigg)^{(1/p-1/2)q}+1\Bigg)\eta_1^q.
\end{align}
which finishes the proof.
\end{proof}

Next, we consider Dantzig selector constraint $\mathcal{B}^{DS}(\eta_2)$.
\begin{proposition}\label{ROP-lp-DS}
Let $r$ be any positive integer and $0<p\leq q\leq 1$.  Let $\hat{X}^{DS}$ be the solution of the Schatten-$p$ minimization (\ref{MatrixSp}) with $\mathcal{B}=\mathcal{B}^{DS}(\eta_2)$.
\begin{itemize}
\item[(1)]
For any rank-$r$ $X$, let $k>1$ with $kr$ be positive integer and $\mathcal{A}$ satisfies $\ell_q$-RUB of order $(k+1)r$ with $C_2/{C_1}<k^{(1/p-/2)q}$, then
\begin{align*}
\|\hat{X}^{DS}&-X\|_{S_2}^q
\leq\frac{2^{q/p+q}r^{(1/p-1/2)q}}{L^q\rho_1^2}\Bigg(\bigg(\frac{1}{k}\bigg)^{(1/p-1/2)q}+1\Bigg)\eta_2^q,
\end{align*}
where $\rho_1=C_1-C_2\big(1/k\big)^{(1/p-1/2)q}$.
\item[(2)]
And for more general matrix $X$, let $k>2^{\frac{2(q-p)}{q(2-p)}}$ with $kr$ be positive integer, and $\mathcal{A}$ satisfies $\ell_q$-RUB of order $(k+1)r$ with $C_2/{C_1}<2^{1-q/p}k^{(1/p-/2)q}$, then
\begin{align*}
\|\hat{X}^{DS}&-X\|_{S_2}^q\\
&\leq\Bigg\{\bigg(\frac{C_22^{2q/p}}{\rho_2}\bigg(\frac{1}{k}\bigg)^{(1/p-1/2)q}+\frac{3}{2}\bigg)
\Bigg(2^{2q/p-1}\bigg(\frac{1}{k}\bigg)^{(1/p-1/2)q}+1\Bigg)
\bigg(\frac{\|X_{-\max(r)}\|_{S_p}}{r^{1/p-1/2}}\bigg)^q\nonumber\\
&\hspace*{12pt}+\frac{2^{2q/p+q+1}r^{(1/p-1/2)q}}{L^q\rho_2^2}\Bigg(2^{q/p-1}
\bigg(\frac{1}{k}\bigg)^{(1/p-1/2)q}+1\Bigg)\eta_2^q,
\end{align*}
where $\rho_2=C_1-C_22^{q/p-1}\big(1/k\big)^{(1/p-1/2)q}$.
\end{itemize}
\end{proposition}

\begin{proof}
We denote $l=\min\{m,n\}$. Let $R=\hat{X}^{DS}-X$. We have the following tube constraint inequality
\begin{align}\label{e3.11}
\|\mathcal{A}^*(\mathcal{A}(R))\|_{S_\infty}
\leq\|\mathcal{A}^*(\mathcal{A}(\hat{X}^{DS})-b)\|_{S_\infty}+\|\mathcal{A}^*(b-\mathcal{A}(X))\|_{S_\infty}
\leq\eta_2+\eta_2=2\eta_2,
\end{align}
instead of (\ref{e3.1}). By Lemma \ref{MatrixConeconstraintinequality}, cone constraint inequality (\ref{e3.2}) still holds.

With the same proof of Theorem \ref{ROP-lp-lq}, we still have (\ref{e3.3}) (or (\ref{e3.3-1})), which gives out a lower bound of $\|\mathcal{A}(R)\|_q^q$.

But the upper bound of $\|\mathcal{A}(R)\|_q^q$ needs a completely different proof from that of Theorem \ref{ROP-lp-lq}. Using (\ref{e3.11}), we have
\begin{align}\label{e3.12}
\|\mathcal{A}(R)\|_q^q&\leq L^{1-q/2}\|\mathcal{A}(R)\|_2^q
=L^{1-q/2}\big(\langle\mathcal{A}(R),\mathcal{A}(R) \rangle\big)^{q/2}\nonumber\\
&\leq L^{1-q/2}\big(\|\mathcal{A}^{*}(\mathcal{A}(R))\|_{S_{\infty}}\|R\|_{S_1}\big)^{q/2}
\leq L^{1-q/2}\big(2\eta_2\|R\|_{S_p}\big)^{q/2}\nonumber\\
&= L^{1-q/2}(2\eta_2)^{q/2}\big(\|R_{\max(r)}\|_{S_p}^p+\|R_{-max(r)}\|_{S_p}^p\big)^{q/(2p)}\nonumber\\
&\leq L^{1-q/2}(2\eta_2)^{q/2}\big(2\|\sigma_{\max(r)}\|_{p}^p+2\|X_{-max(r)}\|_{S_p}^p\big)^{q/(2p)}\nonumber\\
&\leq
\Bigg\{L^{2-q}2^{q/p+q}\eta_2^qr^{(1/p-1/2)q}
\Bigg(\|\sigma_{\max(r)}\|_{2}^p+\bigg(\frac{\|X_{-max(r)}\|_{S_p}}{r^{1/p-1/2}}\bigg)^{p}\Bigg)^{q/p}\Bigg\}^{1/2}\\
&\leq
\begin{cases}
\Bigg\{L^{2-q}2^{q/p+q}\eta_2^qr^{(1/p-1/2)q}
\|R_{T_{01}}\|_{S_2}^q\Bigg\}^{1/2}, &\text{if}~X ~\text{is~rank}~r,\nonumber\\
\Bigg\{L^{2-q}2^{2q/p+q-1}\eta_2^qr^{(1/p-1/2)q}
\Bigg(\|R_{T_{01}}\|_{S_2}^q+\bigg(\frac{\|X_{-max(r)}\|_{S_p}}{r^{1/p-1/2}}\bigg)^{q}\Bigg)\Bigg\}^{1/2}, & ~\text{otherwise}.
\end{cases}
\end{align}

Combining the lower bound (\ref{e3.3-1}) (or \ref{e3.3}) with the upper bound (\ref{e3.12}), we get
\begin{align}\label{e3.13-1}
\Bigg\{L\bigg(C_1-C_2\bigg(\frac{1}{k}\bigg)^{(1/p-1/2)q}\bigg)\|R_{T_{01}}\|_{S_2}^q\Bigg\}^2
\leq L^{2-q}2^{q/p+q}r^{(1/p-1/2)^q}\eta_2^q\|R_{T_{01}}\|_{S_2}^q
\end{align}
or
\begin{align}\label{e3.13}
\Bigg\{L\bigg(C_1-C_22^{q/p-1}\bigg(\frac{1}{k}\bigg)^{(1/p-1/2)q}\bigg)\|R_{T_{01}}\|_{S_2}^q
-C_2L2^{2q/p-1}\bigg(\frac{1}{k}\bigg)^{(1/p-1/2)q}\bigg(\frac{\|X_{-\max(r)}\|_{S_p}}{r^{1/p-1/2}}\bigg)^q\Bigg\}_{+}^2\nonumber\\
\leq L^{2-q}2^{2q/p+q-1}r^{(1/p-1/2)^q}\eta_2^q
\Bigg(\|R_{T_{01}}\|_{S_2}^q+\bigg(\frac{\|X_{-max(r)}\|_{S_p}}{r^{1/p-1/2}}\bigg)^{q}\Bigg),
\end{align}
where $(x)_+=\max\{x,0\}$.

First, we consider (\ref{e3.13}). Let $x=\|R_{T_{01}}\|_{S_2}^q$ and $y=\big(r^{-1/p+1/2}\|X_{-max(r)}\|_{S_p}^p\big)^{q}$.
Note that
$$
\rho=
C_1-C_22^{q/p-1}\bigg(\frac{1}{k}\bigg)^{(1/p-1/2)q}>0.
$$
When
\begin{align*}
\|R_{T_{01}}\|_{S_2}^q\geq \frac{C_22^{2q/p-1}}{\rho}\bigg(\frac{1}{k}\bigg)^{(1/p-1/2)q}\bigg(\frac{\|X_{-\max(r)}\|_{S_p}}{r^{1/p-1/2}}\bigg)^q,
\end{align*}
we have
\begin{align*}
L^2\rho^2&x^2-\bigg(C_2L^2\rho2^{2q/p}\bigg(\frac{1}{k}\bigg)^{(1/p-1/2)q}y+L^{2-q}2^{2q/p+q-1}r^{(1/p-1/2)q}\eta_2^q\bigg)x\\
&-L^{2-q}2^{2q/p+q-1}r^{(1/p-1/2)q}\eta_2^qy\leq 0.
\end{align*}
Note that for the second order inequality $ax^2-bx-c\leq 0$ for $a,b,c>0$, we have
$$
x\leq\frac{b+\sqrt{b^2+4ac}}{2a}\leq\frac{b}{a}+\sqrt{\frac{c}{a}}.
$$
Hence we can get an upper bound of $x$ as follows
\begin{align}\label{e3.15}
x&\leq
\frac{C_2L^2\rho2^{2q/p}}{L^2\rho^2}\bigg(\frac{1}{k}\bigg)^{(1/p-1/2)q}y
+\frac{L^{2-q}2^{2q/p+q-1}r^{(1/p-1/2)q}}{L^2\rho^2}\eta_2^q\nonumber\\
&\hspace*{12pt}+\sqrt{\frac{L^{2-q}2^{2q/p+q-1}r^{(1/p-1/2)q}\eta_2^qy}{L^2\rho^2}}\nonumber\\
&\leq\frac{C_22^{2q/p}}{\rho}\bigg(\frac{1}{k}\bigg)^{(1/p-1/2)q}y
+\frac{2^{2q/p+q-1}r^{(1/p-1/2)q}}{L^q\rho^2}\eta_2^q\nonumber\\
&\hspace*{12pt}+\frac{1}{2}\bigg(\frac{2^{2q/p+q-1}r^{(1/p-1/2)q}}{L^q\rho^2}\eta_2^q+y\bigg)\nonumber\\
&\leq\bigg(\frac{C_22^{2q/p}}{\rho}\bigg(\frac{1}{k}\bigg)^{(1/p-1/2)q}+\frac{1}{2}\bigg)y
+\frac{2^{2q/p+q+1}r^{(1/p-1/2)q}}{L^q\rho^2}\eta_2^q.
\end{align}
Hence whenever
\begin{align*}
\|R_{T_{01}}\|_{S_2}^q\geq \frac{C_22^{2q/p-1}}{\rho}\bigg(\frac{1}{k}\bigg)^{(1/p-1/2)q}\bigg(\frac{\|X_{-\max(r)}\|_{S_p}}{r^{1/p-1/2}}\bigg)^q
\end{align*}
or not, we always have
\begin{align}\label{e3.16}
\|R_{T_{01}}\|_{S_2}^q&\leq \max\Bigg\{\frac{C_22^{2q/p-1}}{\rho}\bigg(\frac{1}{k}\bigg)^{(1/p-1/2)q}
\bigg(\frac{\|X_{-\max(r)}\|_{S_p}}{r^{1/p-1/2}}\bigg)^q,\nonumber\\
&\hspace{12pt}\bigg(\frac{C_22^{2q/p}}{\rho}\bigg(\frac{1}{k}\bigg)^{(1/p-1/2)q}+\frac{1}{2}\bigg)
\bigg(\frac{\|X_{-\max(r)}\|_{S_p}}{r^{1/p-1/2}}\bigg)^q
+\frac{2^{2q/p+q+1}r^{(1/p-1/2)q}}{L^q\rho^2}\eta_2^q\Bigg\}\nonumber\\
&\leq \bigg(\frac{C_22^{2q/p}}{\rho}\bigg(\frac{1}{k}\bigg)^{(1/p-1/2)q}+\frac{1}{2}\bigg)
\bigg(\frac{\|X_{-\max(r)}\|_{S_p}}{r^{1/p-1/2}}\bigg)^q
+\frac{2^{2q/p+q+1}r^{(1/p-1/2)q}}{L^q\rho^2}\eta_2^q.
\end{align}

Then, we consider (\ref{e3.13-1}). Note that
$$
\rho=C_1-C_2\bigg(\frac{1}{k}\bigg)^{(1/p-1/2)q}>0.
$$
Therefore
\begin{align}\label{e3.16-1}
\|R_{T_{01}}\|_{S_2}^q
\leq \frac{2^{q/p+q}r^{(1/p-1/2)q}}{L^q\rho^2}\eta_2^q.
\end{align}

Note that (\ref{e3.9}) and (\ref{e3.9-1}) still holds.

Then substituting (\ref{e3.16}) into (\ref{e3.9}), we have
\begin{align}\label{e3.17}
\|R\|_{S_2}^q&\leq
\Bigg\{\bigg(\frac{C_22^{2q/p}}{\rho}\bigg(\frac{1}{k}\bigg)^{(1/p-1/2)q}+\frac{1}{2}\bigg)
\Bigg(2^{q/p-1}\bigg(\frac{p}{2}\bigg)^{q/2}
\bigg(\frac{2-p}{2}\bigg)^{(1/p-1/2)q}\bigg(\frac{1}{k}\bigg)^{(1/p-1/2)q}+1\Bigg)
\nonumber\\
&\hspace*{12pt}+2^{2q/p-1}\bigg(\frac{p}{2}\bigg)^{q/2}
\bigg(\frac{2-p}{2}\bigg)^{(1/p-1/2)q}\bigg(\frac{1}{k}\bigg)^{(1/p-1/2)q}\Bigg\}
\bigg(\frac{\|X_{-\max(r)}\|_{S_p}}{r^{1/p-1/2}}\bigg)^q\nonumber\\
&\hspace*{12pt}+ \frac{2^{2q/p+q+1}r^{(1/p-1/2)q}}{L^q\rho^2}\Bigg(2^{q/p-1}\bigg(\frac{p}{2}\bigg)^{q/2}
\bigg(\frac{2-p}{2}\bigg)^{(1/p-1/2)q}\bigg(\frac{1}{k}\bigg)^{(1/p-1/2)q}+1\Bigg)\eta_2^q\nonumber\\
&\leq\Bigg\{\bigg(\frac{C_22^{2q/p}}{\rho}\bigg(\frac{1}{k}\bigg)^{(1/p-1/2)q}+\frac{3}{2}\bigg)
\Bigg(2^{2q/p-1}\bigg(\frac{p}{2}\bigg)^{q/2}
\bigg(\frac{2-p}{2}\bigg)^{(1/p-1/2)q}\bigg(\frac{1}{k}\bigg)^{(1/p-1/2)q}+1\Bigg)
\nonumber\\
&\hspace*{12pt}\times
\bigg(\frac{\|X_{-\max(r)}\|_{S_p}}{r^{1/p-1/2}}\bigg)^q\nonumber\\
&\hspace*{12pt}+ \frac{2^{2q/p+q+1}r^{(1/p-1/2)q}}{L^q\rho^2}\Bigg(2^{q/p-1}
\bigg(\frac{p}{2}\bigg)^{q/2}
\bigg(\frac{2-p}{2}\bigg)^{(1/p-1/2)q}
\bigg(\frac{1}{k}\bigg)^{(1/p-1/2)q}+1\Bigg)\eta_2^q\nonumber\\
&\leq\Bigg\{\bigg(\frac{C_22^{2q/p}}{\rho}\bigg(\frac{1}{k}\bigg)^{(1/p-1/2)q}+\frac{3}{2}\bigg)
\Bigg(2^{2q/p-1}\bigg(\frac{1}{k}\bigg)^{(1/p-1/2)q}+1\Bigg)
\bigg(\frac{\|X_{-\max(r)}\|_{S_p}}{r^{1/p-1/2}}\bigg)^q\nonumber\\
&\hspace*{12pt}+ \frac{2^{2q/p+q+1}r^{(1/p-1/2)q}}{L^q\rho^2}\Bigg(2^{q/p-1}
\bigg(\frac{1}{k}\bigg)^{(1/p-1/2)q}+1\Bigg)\eta_2^q,
\end{align}
where the last inequality follows from $0<p/2\leq 1$ and $(1/p-1/2)q>0$.

And substituting (\ref{e3.16-1}) into (\ref{e3.9-1}), we have
\begin{align}\label{e3.17-1}
\|R\|_{S_2}^q&\leq
\frac{2^{q/p+q}r^{(1/p-1/2)q}}{L^q\rho^2}\Bigg(\bigg(\frac{1}{k}\bigg)^{(1/p-1/2)q}+1\Bigg)\eta_2^q.
\end{align}
which finishes the proof.
\end{proof}

\begin{remark}
What should be noticed is that Zhang, Huang and Zhang \cite{ZHZ2013} also considered matrix recovery via Schatten-$p$ minimization with $\mathcal{B}=\mathcal{B}^{\ell_2}(\eta)$-$\ell_2$ norm constraint, under the $\ell_p$-RUB condition.
The problem (\ref{MatrixSp}) considered in Proposition \ref{ROP-lp-lq} and Theorem \ref{ROP-lp-lq-DS} are different from that of they considered.
\end{remark}

If $\mathcal{A}:\mathbb{S}^{m}\rightarrow \mathbb{R}^L$ is a SROP, we consider Schatten-$p$ minimization (\ref{SROP-MatrixSp}) instead of (\ref{MatrixSp}) and have the following results.
\begin{corollary}\label{SROP-lp-lq-DS}
Let $r$ be any positive integer, $0<p\leq q\leq 1$ and $\tilde{L}=\lfloor L/2 \rfloor$.  Let $\hat{X}^{\ell_q}$ be the solution of the Schatten-$p$ minimization (\ref{SROP-MatrixSp}) with $\mathcal{B}=\mathcal{B}^{\ell_q}(\eta_1)\cap \mathcal{B}^{\ell_q}(\eta_2)$.
\begin{itemize}
\item[(1)]
For any symmetric rank-$r$ $X$, let $k>1$ with $kr$ be positive integer, and $\tilde{\mathcal{A}}$ satisfies $\ell_q$-RUB of order $(k+1)r$ with $C_2/{C_1}<k^{(1/p-/2)q}$, then
\begin{align*}
\|\hat{X}&-X\|_{S_2}^q\leq
\Bigg(\bigg(\frac{1}{k}\bigg)^{(1/p-1/2)q}+1\Bigg)\frac{1}{\rho_1 \tilde{L}^{q}}
\min\bigg\{\frac{2}{\tilde{L}^{1-2q}}\eta_1^q,\frac{2^{q/p+q}}{\rho_1}r^{(1/p-1/2)q}\eta_2^q\bigg\},
\end{align*}
where $\rho_1=C_1-C_2\big(1/k\big)^{(1/p-1/2)q}$.
\item[(2)]
And for more general symmetric matrix $X$, let $k>2^{\frac{2(q-p)}{q(2-p)}}$ with $kr$ be positive integer, and $\tilde{\mathcal{A}}$ satisfies $\ell_q$-RUB of order $(k+1)r$ with $C_2/{C_1}<2^{1-q/p}k^{(1/p-/2)q}$, then
\begin{align*}
\|\hat{X}&-X\|_{S_2}^q\\
&\leq \Bigg(\frac{C_22^{2q/p-1}}{\rho_2}\bigg(\frac{1}{k}\bigg)^{(1/p-1/2)q}+1\Bigg)
\Bigg(2^{2q/p-1}\bigg(\frac{1}{k}\bigg)^{(1/p-1/2)q}+1\Bigg)\bigg(\frac{\|X_{-\max(r)}\|_{S_p}}{r^{1/p-1/2}}\bigg)^q\nonumber\\
&\hspace*{12pt}+\Bigg(2^{q/p-1}\bigg(\frac{1}{k}\bigg)^{(1/p-1/2)q}+1\Bigg)\frac{1}{\rho_2 \tilde{L}^{q}}
\min\bigg\{\frac{2}{\tilde{L}^{1-2q}}\eta_1^q,\frac{2^{2q/p+q+1}}{\rho_1}r^{(1/p-1/2)q}\eta_2^q\bigg\},
\end{align*}
where $\rho_2=C_1-C_22^{q/p-1}\big(1/k\big)^{(1/p-1/2)q}$.
\end{itemize}
\end{corollary}

\section{Robust Null Space Property and $\ell_q$-RUB\label{s3}}
\hskip\parindent

In this section, we will investigate the relationship between the RUB and robust rank null space property.

The robust null space property with $\ell_2$ bound $\|Ax\|_2$ was first introduced by Sun in \cite{S2011}, which is called sparse approximation property. And this name was first used by Foucart and Rauhut in \cite{FR2013}. And they also introduced the robust null space property with Dantzig selector bound $\|A^*Ax\|_{\infty}$. In \cite{FR2013}, Foucart and Rauhut also introduced the robust rank null space property.

Inspired by \cite{S2011}, we introduce $(\ell_t,\ell_p)$-robust rank null space property as follows.

\begin{definition}\label{MatrixRNSP}
Let $0<p,t\leq\infty$ and $0<q<\infty$. A linear map $\mathcal{A}$ satisfies $(\ell_t,\ell_p)$-robust rank null space property of order $r$ for $\ell_q$ bound with constants $D$ and $\beta$ if
\begin{align}\label{MatrixlpRNSP}
\|X_{\max(r)}\|_{S_t}^p \leq D\|\mathcal{A}(X)\|_q^p + \beta\frac{\|X_{-\max(r)}\|_{S_p}^p}{r^{(1/p-1/t)p}}
\end{align}
holds for all $X\in\mathbb{R}^{m\times n}$.

And a linear map $\mathcal{A}$ satisfies the $(\ell_t,\ell_p)$ robust rank null space property of order $r$ for Dantzig selector bound  with constants $D$ and $\beta$ if
\begin{align}\label{MatrixDSRNSP}
\|X_{\max(r)}\|_{S_t}^p \leq D\|\mathcal{A}^{*}\mathcal{A}(X)\|_{S_\infty}^p +\beta\frac{\|X_{-\max(r)}\|_{S_p}^p}{r^{(1/p-1/t)p}}
\end{align}
holds for all $X\in\mathbb{R}^{m\times n}$.
\end{definition}

By similar proofs as that of Theorems \ref{ROP-lp-lq} and \ref{ROP-lp-DS}, we have the following theorem, which implies that  $(\ell_2,\ell_p)$-robust rank null space property of order $r$ can be induced from $\ell_q$-RUB of order $(k+1)r$ for any fixed $k>1$ with $kr\in\mathbb{Z}_+$. We omit the details here.
\begin{theorem}\label{RUB-RNSP}
Let $r\in\mathbb{Z}_+$, and $k>1$ with $kr\in\mathbb{Z}_+$.
\begin{itemize}
\item [(1)]
Suppose that $\mathcal{A}$ satisfies $\ell_q$-RUB of order $(k+1)r$ with $C_2/{C_1}<k^{(1/p-1/2)q}$ for any $0<p\leq q\leq 1$, then
\begin{align*}
\|X_{\max(r)}\|_{S_2}^p&\leq \frac{1}{(C_1L)^{p/q}}\|\mathcal{A}(X)\|_q^p + \bigg(\frac{C_2}{C_1k^{(1/p-1/2)q}}\bigg)^{p/q}\frac{\|X_{-\max(r)}\|_{S_p}^p}{r^{(1/p-1/2)p}}\\
&=:D_1\|\mathcal{A}(X)\|_q^p+\beta_1\frac{\|X_{-\max(r)}\|_{S_p}^p}{r^{(1/p-1/2)p}},
\end{align*}
i.e., $\mathcal{A}$ satisfies $(\ell_2,\ell_p)$ robust rank null space property of order $r$ for $\ell_q$ bound with constant pair $(D_1,\beta_1)$.
\item[(2)]
Suppose that $\mathcal{A}$ satisfies $\ell_q$-RUB of order $(k+1)r$ with $C_2/{C_1}<k^{(1/p-1/2)q}/4$ for any $0<p\leq q\leq 1$, then
\begin{align*}
\|X_{\max(r)}\|_{S_2}^p&\leq \bigg(\frac{2^{q/p+1}}{C_1^{2p/q}L^p}\bigg)^{p/q}
r^{(1/p-1/2)p}\|\mathcal{A}^*\mathcal{A}(X)\|_{S_\infty}^p
+\bigg(\frac{2C_2}{C_1k^{(1/p-1/2)q}}+\frac{1}{2}\bigg)^{p/q}\frac{\|X_{-\max(r)}\|_{S_p}^p}{r^{(1/p-1/2)p}}\\
&:=D_2\|\mathcal{A}^{*}\mathcal{A}(X)\|_{S_\infty}^p +\beta_2 \frac{\|X_{-\max(r)}\|_{S_p}^p}{r^{(1/p-1/2)p}}
\end{align*}
holds for all $X\in\mathbb{R}^{m\times n}$, i.e., $\mathcal{A}$ satisfies $(\ell_2,\ell_p)$ robust rank null space property of order $r$ for Dantzig selector bound with constant pair $(D_2,\beta_2)$.
\end{itemize}
\end{theorem}

\begin{remark}
Theorem \ref{RUB-RNSP} also holds for vector case. We only need to let $X$ be any diagonal matrix, and let each element $A_j$ of linear map $\mathcal{A}$ be diagonal matrix.
\end{remark}

\section{Recovery via Least $q$ Minimization \label{s4}}
\hskip\parindent

In this section, we will consider the matrix's recovery via least $q$ minimization (\ref{DensityLeastq}). We show that
the $\ell_q$-RUB of order $(k+1)r$ with $C_2/{C_1}<k^{(1/p-1/2)q}$ for $0<p\leq q\leq 1$ is also sufficient to recover any matrix $X\in\mathbb{R}^{m\times n}$. We also extend our result to Gaussian design distribution and assert that $L\geq Cr(m+n)$ measurements is sufficient to guarantee all rank-$r$ matrices' stable recovery with high probability.

\subsection{Stable Recovery via Least $q$ Minimization \label{s4.1}}
\quad

\begin{theorem}\label{StableDensityLeastq}
Let $r\in\mathbb{Z}_+$, and $k>1$ with $kr\in\mathbb{Z}_+$.
Suppose that $\mathcal{A}$ satisfies $\ell_q$-RUB of order $(k+1)r$ with $C_2/{C_1}<k^{(1/p-1/2)q}$ for any $0<p\leq q\leq 1$. Let $\hat{X}$ be the solution of
\begin{align*}
\min_{Y\in\mathbb{R}^{m\times n}}~\|\mathcal{A}(Y)-b\|_q^q~~\text{subject~ to}~~\Big(\text{tr}\big((Y^{T}Y)^{p/2}\big)\Big)^{1/p}=1
\end{align*}
where $b=\mathcal{A}(X)+z$, then
\begin{align*}
\|\hat{X}-X\|_{S_2}^p
&\leq\frac{2(1+\beta_1)^2}{1-\beta_1} \frac{\|X_{-\max(r)}\|_p^p}{r^{(1/p-1/2)p}}
+\frac{2^{p/q}(3+\beta_1)D_1}{1-\beta_1}\|z\|_q^p,
\end{align*}
where $D_1= (C_1L)^{-p/q}>0$ and $0<\beta_1=\big(C_2k^{-(1/p-1/2)q}/{C_1}\big)^{p/q}<1$.
\end{theorem}

Theorem \ref{StableDensityLeastq} can be deduced from the following stronger result and Theorem \ref{RUB-RNSP}.

\begin{theorem}\label{RNSP-stabelrecovery}
Let $0<p\leq t\leq\infty$.
\begin{itemize}
\item[(1)]
If $\mathcal{A}: \mathbb{R}^{m\times n}\rightarrow\mathbb{R}^L$ satisfies the $(\ell_t,\ell_p)$-robust rank null space property of order $r$ for $\ell_q$ bound with constants $D_1>0$ and $0<\beta_1<1$, then
\begin{align*}
\|Y-X\|_{S_t}^p\leq \frac{(1+\beta_1)^2}{1-\beta_1} \frac{1}{r^{(1/p-1/t)p}}\big(\|Y\|_{S_p}^p-\|X\|_{S_p}^p+2\|X_{-\max(r)}\|_p^p\big)
+\frac{(3+\beta_1)D_1}{1-\beta_1}\|\mathcal{A}(Y-X)\|_q^p.
\end{align*}
\item[(2)]
If $\mathcal{A}: \mathbb{R}^{m\times n}\rightarrow\mathbb{R}^L$ satisfies the $(\ell_t,\ell_p)$-robust rank null space property of order $r$ for Dantzig selector bound with constants $D_2>0$ and $0<\beta_2<1$, then
\begin{align*}
\|Y-X\|_{S_t}^p\leq \frac{(1+\beta_2)^2}{1-\beta_2} \frac{1}{r^{(1/p-1/t)p}}\big(\|Y\|_{S_p}^p-\|X\|_{S_p}^p+2\|X_{-\max(r)}\|_p^p\big)
+\frac{(3+\beta_2)D_2}{1-\beta_2}\|\mathcal{A}^*\mathcal{A}(Y-X)\|_{S_\infty}^p.
\end{align*}
\end{itemize}
\end{theorem}

The proof requires some auxiliary lemmas. We start with a matrix version of Stechkin's bound. It follows immediately from \cite[Proposition 2.3 of Chapter 2]{FR2013}.

\begin{lemma}\label{Stechkinbound}
Let $X\in\mathbb{R}^{m\times n}$ and $0<r\leq \min\{m,n\}$. Then for any $0<p\leq t\leq\infty$,
$$
\|X_{-\max(r)}\|_{S_t}\leq\frac{\|X\|_{S_p}}{r^{1/p-1/t}}.
$$
\end{lemma}

In order to get a similar cone constraint for matrix's  Schatten-$p$ norm, we also need the following lemma which was given by Yue and So in \cite{YS2016} and Audenaert \cite{A2014}.

\begin{lemma}\label{perturbationinequality}
Let $X, Y\in\mathbb{R}^{m\times n}$ be given matrices. Suppose that $f: \mathbb{R}_{+}\rightarrow \mathbb{R}_{+}$is a concave function satisfying $f(0)=0$. Then, for any $k\in\{1,\ldots,\min\{m,n\}\}$, we have
$$\sum_{j=1}^k|f(\sigma_j(X))-f(\sigma_j(Y))|\leq\sum_{j=1}^kf(\sigma_j(X-Y)).$$
\end{lemma}

\begin{lemma}\label{RNSPlemma}
Let $0<p\leq t\leq\infty$.
\begin{itemize}
\item[(1)]
If $\mathcal{A}: \mathbb{R}^{m\times n}\rightarrow\mathbb{R}^L$ satisfies the $(\ell_t,\ell_p)$-robust rank null space property of order $r$ for $l_q$ bound with constants $D_1>0$ and $0<\beta_1<1$, then
\begin{align*}
\|Y-X\|_{S_p}^p\leq \frac{2D_1}{1-\beta_1}r^{(1/p-1/t)p}\|\mathcal{A}(Y-X)\|_q^p
+\frac{1+\beta_1}{1-\beta_1}\big(\|Y\|_{S_p}^p-\|X\|_{S_p}^p+2\|X_{-\max(r)}\|_p^p\big).
\end{align*}
\item[(2)]
If $\mathcal{A}: \mathbb{R}^{m\times n}\rightarrow\mathbb{R}^L$ satisfies the $(\ell_t,\ell_p)$-robust rank null space property of order $r$ for Dantzig selector bound with constants $D_2>0$ and $0<\beta_2<1$, then
\begin{align*}
\|Y-X\|_{S_p}^p\leq \frac{2D_2}{1-\beta_2}r^{(1/p-1/t)p}\|\mathcal{A}^*\mathcal{A}(Y-X)\|_{S_\infty}^p
+\frac{1+\beta_2}{1-\beta_2}\big(\|Y\|_{S_p}^p-\|X\|_{S_p}^p+2\|X_{-\max(r)}\|_p^p\big).
\end{align*}
\end{itemize}
\end{lemma}
\begin{proof}
Our proof follows the idea of \cite[Section 4.3]{FR2013}. Let $R=Y-X$. Lemma \ref{orthogonaldecomposition} states that
\begin{align*}
\|R\|_{S_p}^p=\|R_{\max(r)}\|_{S_p}^p+\|R_{-\max(r)}\|_{S_p}^p\leq r^{(1/p-1/t)p}\|R_{\max(r)}\|_{S_t}^p+\|R_{-\max(r)}\|_{S_p}^p.
\end{align*}

Applying the $(l_t,l_p)$-robust rank null space property of $\mathcal{A}$, we obtain that
\begin{align}\label{e5.1}
\|R_{\max(r)}\|_{S_t}^p\leq D_1\|\mathcal{A}(R)\|_q^p + \beta_1\frac{\|R_{-\max(r)}\|_{S_p}^p}{r^{(1/p-1/t)p}}.
\end{align}
Therefore, we have
\begin{align}\label{e5.2}
\|R\|_{S_p}^p\leq D_1r^{(1/p-1/t)p}\|\mathcal{A}(R)\|_q^p+(1+\beta_1)\|R_{-\max(r)}\|_{S_p}^p.
\end{align}

Next, we estimate the upper bound of $\|R_{-\max(r)}\|_{S_p}^p$. Since $x\rightarrow|x|^p$ is concave on $\mathbb{R}_+$ for any $p\in(0, 1]$, by taking $f(\cdot)=(\cdot)^p$ in Lemma \ref{perturbationinequality}, we immediately obtain
\begin{align*}
\|Y\|_{S_p}^p&=\|\sigma(Y)_{T}\|_p^p+\|\sigma(Y)_{T^c}\|_p^p\\
&\geq\|\sigma(X)_{T}\|_p^p-\|\sigma(R)_{T}\|_p^p+\|\sigma(R)_{T^c}\|_p^p-\|\sigma(X)_{T^c}\|_p^p,
\end{align*}
where $T=\text{supp}(\sigma(X)_{\max(r)})$.
By rearranging the terms in the above inequality, we get
\begin{align*}
\|R_{-\max(r)}\|_{S_p}^p&\leq\|\sigma(R)_{T^c}\|_p^p
\leq\|Y\|_{S_p}^p-\|\sigma(X)_{T}\|_p^p+\|\sigma(X)_{T^c}\|_p^p+\|\sigma(R)_{T}\|_p^p\\
&\leq\big(\|Y\|_{S_p}^p-\|X\|_{S_p}^p+2\|X_{-\max(r)}\|_p^p\big)+\|R_{\max(r)}\|_{S_p}^p\\
&\leq\big(\|Y\|_{S_p}^p-\|X\|_{S_p}^p+2\|X_{-\max(r)}\|_p^p\big)+r^{(1/p-1/t)p}\|R_{\max(r)}\|_{S_t}^p\\
&\leq\big(\|Y\|_{S_p}^p-\|X\|_{S_p}^p+2\|X_{-\max(r)}\|_p^p\big)
+D_1r^{(1/p-1/t)p}\|\mathcal{A}(R)\|_q^p+\beta_1\|R_{-\max(r)}\|_{S_p}^p,
\end{align*}
where the last line follows from (\ref{e5.1}). Note that $0<\beta_1<1$. Therefore,
\begin{align}\label{e5.3}
\|R_{-\max(r)}\|_{S_p}^p\leq\frac{1}{1-\beta_1}\big(\|Y\|_{S_p}^p-\|X\|_{S_p}^p+2\|X_{-\max(r)}\|_p^p\big)
+\frac{D_1}{1-\beta_1}r^{(1/p-1/t)p}\|\mathcal{A}(R)\|_q^p.
\end{align}

And substituting (\ref{e5.3}) into (\ref{e5.2}), we obtain
\begin{align*}
\|R\|_{S_p}^p&\leq (1+\beta_1)\bigg(\frac{1}{1-\beta_1}\big(\|Y\|_{S_p}^p-\|X\|_{S_p}^p+2\|X_{-\max(r)}\|_p^p\big)
+\frac{D_1}{1-\beta_1}r^{(1/p-1/t)p}\|\mathcal{A}(R)\|_q^p\bigg)\\
&\hspace*{12pt}+D_1r^{(1/p-1/t)p}\|\mathcal{A}(R)\|_q^p\\
&=\frac{2D_1}{1-\beta_1}r^{(1/p-1/t)p}\|\mathcal{A}(R)\|_q^p
+\frac{1+\beta_1}{1-\beta_1}\big(\|Y\|_{S_p}^p-\|X\|_{S_p}^p+2\|X_{-\max(r)}\|_p^p\big),
\end{align*}
which finishes the proof of item (1).

If $\mathcal{A}$ satisfies the $(\ell_t,\ell_p)$-robust rank null space property for Dantzig selector bound, we replace (\ref{e5.1}) with
\begin{align}\label{e5.4}
\|R_{\max(r)}\|_{S_t}^p\leq D_2\|\mathcal{A}^*\mathcal{A}(R)\|_{S_\infty}^p + \beta_2\frac{\|R_{-\max(r)}\|_{S_p}^p}{r^{(1/p-1/t)p}}.
\end{align}
Then by a similar proof as item (1), we get
\begin{align*}
\|R\|_{S_p}^p
&=\frac{2D_2}{1-\beta_2}r^{(1/p-1/t)p}\|\mathcal{A}^*\mathcal{A}(R)\|_{S_\infty}^p
+\frac{1+\beta_2}{1-\beta_2}\big(\|Y\|_{S_p}^p-\|X\|_{S_p}^p+2\|X_{-\max(r)}\|_p^p\big),
\end{align*}
which finishes the proof of item (2).
\end{proof}

\begin{proof}[Proof of Theorem \ref{RNSP-stabelrecovery}]
Let $R=Y-X$, then we have
\begin{align}\label{e5.5}
\|Y-X\|_{S_t}^p=\big(\|R_{\max(r)}\|_{S_t}^t+\|R_{-\max(r)}\|_{S_t}^t\big)^{p/t}
\leq\|R_{\max(r)}\|_{S_t}^p+\|R_{-\max(r)}\|_{S_t}^p,
\end{align}
where the inequality follows from $0<p\leq t$. By the $(\ell_t,\ell_p)$ robust rank null space property of $\mathcal{A}$ with respect to $l_p$,
\begin{align}\label{e5.6}
\|R_{\max(r)}\|_{S_t}^p\leq D_1\|\mathcal{A}(R)\|_q^p+\beta_1\frac{\|R_{-\max(r)}\|_{S_p}^p}{r^{(1/p-1/t)p}}
\leq D_1\|\mathcal{A}(R)\|_q^p+\beta_1\frac{\|R\|_{S_p}^p}{r^{(1/p-1/t)p}}.
\end{align}
Next we estimate $\|R_{-\max(r)}\|_{S_t}^p$. Lemma \ref{Stechkinbound} implies that
\begin{align}\label{e5.7}
\|R_{-\max(r)}\|_{S_t}^p\leq\bigg(\frac{\|R\|_{S_p}}{r^{1/p-1/t}}\bigg)^p.
\end{align}
Then substituting (\ref{e5.6}) and (\ref{e5.7}) into (\ref{e5.5}) yields
\begin{align}\label{e5.8}
\|Y-X\|_{S_t}^p\leq D_1\|\mathcal{A}(R)\|_q^p+
(1+\beta_1)\frac{\|R\|_{S_p}^p}{r^{(1/p-1/t)p}}.
\end{align}
An application of Lemma \ref{RNSPlemma}, we obtain
\begin{align*}
\|Y-X\|_{S_t}^p&\leq
\frac{(1+\beta_1)}{r^{(1/p-1/t)p}}\bigg(\frac{2D_1}{1-\beta_1}r^{(1/p-1/t)p}\|\mathcal{A}(R)\|_q^p
+\frac{1+\beta_1}{1-\beta_1}\big(\|Y\|_{S_p}^p-\|X\|_{S_p}^p+2\|X_{-\max(r)}\|_p^p\big)\bigg)\\
&\hspace*{12pt}+D_1\|\mathcal{A}(R)\|_q^p\\
&=\frac{(1+\beta_1)^2}{1-\beta_1} \frac{1}{r^{(1/p-1/t)p}}\big(\|Y\|_{S_p}^p-\|X\|_{S_p}^p+2\|X_{-\max(r)}\|_p^p\big)
+\frac{(3+\beta_1)D_1}{1-\beta_1}\|\mathcal{A}(R)\|_q^p.
\end{align*}

And if $\mathcal{A}$ satisfies the $(\ell_t,\ell_p)$-robust rank null space property for Dantzig selector bound, then by a similar proof, we have
\begin{align*}
\|Y-X\|_{S_t}^p
&\leq\frac{(1+\beta_2)^2}{1-\beta_2} \frac{1}{r^{(1/p-1/t)p}}\big(\|Y\|_{S_p}^p-\|X\|_{S_p}^p+2\|X_{-\max(r)}\|_p^p\big)
+\frac{(3+\beta_2)D_2}{1-\beta_2}\|\mathcal{A}^*\mathcal{A}(R)\|_{S_\infty}^p.
\end{align*}
\end{proof}

Now, we begin to prove Theorem \ref{StableDensityLeastq}.
\begin{proof}[Proof of Theorem \ref{StableDensityLeastq}]
Note that $\mathcal{A}$ satisfies $\ell_q$-RUB of order $(k+1)r$ with $C_2/{C_1}<k^{(1/p-1/2)q}$ for any $0<p\leq q\leq 1$. Therefore Theorem \ref{RUB-RNSP} implies the validity of the
$(\ell_2,\ell_q)$ robust rank null space property for $\ell_p$ bound with parameters $D_1= (C_1L)^{-p/q}>0$ and $0<\beta_1=\big(C_2k^{-(1/p-1/2)q}/{C_1}\big)^{p/q}<1$.
Then Theorem \ref{RNSP-stabelrecovery} leads to us that
\begin{align}\label{e5.9}
\|\hat{X}-X\|_{S_2}^p
&\leq\frac{(1+\beta_1)^2}{1-\beta_1} \frac{1}{r^{(1/p-1/2)p}}\big(\|\hat{X}\|_{S_p}^p-\|X\|_{S_p}^p+2\|X_{-\max(r)}\|_p^p\big)\nonumber\\
&\hspace*{12pt}+\frac{(3+\beta_1)D_1}{1-\beta_1}\|\mathcal{A}(\hat{X}-X)\|_q^p.
\end{align}

We consider recovering density matrix. This property assures
\begin{align}\label{e5.10}
\|\hat{X}\|_{S_p}^p-\|X\|_{S_p}^p
=\text{tr}\big((\hat{X}^{T}\hat{X})^{p/2}\big)-\text{tr}\big((X^{T}X)^{p/2}\big)=0.
\end{align}
And
\begin{align}\label{e5.11}
\|\mathcal{A}(\hat{X}-X)\|_q^p\leq\big(\|\mathcal{A}(\hat{X})-b\|_q^q +\|b-\mathcal{A}(X)\|_q^q\big)^{p/q}
\leq\big(\|z\|_q^q+\|z\|_q^q\big)^{p/q}=2^{p/q}\|z\|_q^p.
\end{align}
Combining (\ref{e5.10}), (\ref{e5.11}) with (\ref{e5.9}), we get
\begin{align*}
\|\hat{X}-X\|_{S_2}^p
&\leq\frac{2(1+\beta_1)^2}{1-\beta_1} \frac{\|X_{-\max(r)}\|_p^p}{r^{(1/p-1/2)p}}
+\frac{2^{p/q}(3+\beta_1)D_1}{1-\beta_1}\|z\|_q^p.
\end{align*}
\end{proof}

\begin{remark}
Kabanava, Kueng, Rauhut et.al. \cite{KKRT2016} considered the least-$q$ minimization for $q\geq1$. Here, we extend it to nonconvex case,i.e., least-$q$ minimization for $0<q<1$.
\end{remark}

\subsection{Matrix Recovery for ROP from Gauss Distribution via LAD \label{s4.2}}
\quad

This subsection aims to show that recovering low-rank matrices through ROP $\mathcal{A}$ from Gaussian distribution $\mathcal{P}$ is still possible. First, we recall Gaussian random variable $\mathcal{A}$.

A linear map $\mathcal{A}: \mathbb{R}^{m\times n}\rightarrow\mathbb{R}^{L}$ is called ROP from distribution $\mathcal{P}$ if $\mathcal{A}$ is defined as in (\ref{ROP2}) with all the entries of $\beta^j$ and $\gamma^j$ independently drawn from the distribution $\mathcal{P}$. The following Theorem \ref{GaussStableGeneralLAD} shows that recovering low-rank matrices through ROP $\mathcal{A}$ from the standard normal distribution $\mathcal{P}$ via least 1 minimization (\ref{DensityLeastq}) is possible.

\begin{theorem}\label{GaussStableGeneralLAD}
Let $r\in\mathbb{Z}_+$.
Suppose that $\mathcal{A}$ is a ROP from the standard normal distribution. Let $\hat{X}$ be the solution of \begin{align*}
\min_{Y\in\mathbb{R}^{m\times n}}~\|\mathcal{A}(Y)-b\|_1~~\text{subject~ to}~~\Big(\text{tr}\big((Y^{T}Y)^{p/2}\big)\Big)^{1/p}=1
\end{align*}
where $b=\mathcal{A}(X)+z$. Then there exist uniform constants $c_1$, $c_2$, $C_3$ and $C_4$ such that,  whenever $L\geq c_1r(m+n)$, $\hat{X}$ obeys
\begin{align*}
\|\hat{X}-X\|_{S_2}^p\leq
C_3 \frac{\|X_{-\max(r)}\|_p^p}{r^{(1/p-1/2)p}}
+\frac{C_4}{L^p}\|z\|_1^p
\end{align*}
with probability at least $1-e^{-c_2 L}$.
\end{theorem}

\begin{proof}
Suppose $k=10$, by \cite[Theorem 2.2]{CZ2015}, we can find a uniform
constant $c_0$ and $c_2$ such that if $L\geq c_0(k+1)r(m+n)$, $\mathcal{A}$ satisfies RUB of order $11r$
and constants $C_1=0.32,~C_2 = 1.01$ with probability at least $1-e^{-c_2 L}$. Hence, we
have $c_1=11c_0$ and $c_2$ such that if $L\geq c_1r(m+n)$, $\mathcal{A}$ satisfies RUB of order $11r$
and constants satisfying $C_2/C_1<10^{1/p-1/2}$ with probability at least $1-e^{-c_2 L}$.
Then it follows from Theorem \ref{StableDensityLeastq} that
\begin{align*}
\|\hat{X}-X\|_{S_2}^p
&\leq\frac{2(1+\beta_1)^2}{1-\beta_1} \frac{\|X_{-\max(r)}\|_p^p}{r^{(1/p-1/2)p}}
+\frac{2^{p/q}(3+\beta_1)D_1}{1-\beta_1}\|z\|_q^p\\
&=C_3 \frac{\|X_{-\max(r)}\|_p^p}{r^{(1/p-1/2)p}}
+\frac{C_4}{L^p}\|z\|_1^p
\end{align*}
holds with probability at least $1-e^{-c_2 L}$, where $C_3,C_4$ are constants depending only on $p$.
\end{proof}

\begin{corollary}
For the PhaseLift introduced in \cite{CESV2013-2015,CSV2013}, we consider
\begin{align}\label{PhaseLift}
\min_{X\in\mathbb{S}^{m}}~\|\tilde{\mathcal{A}}(X)-\tilde{b}\|_1~~\text{subject~ to}~~X\succeq O~\text{and}~\text{tr}(X)=1.
\end{align}
Then Theorem \ref{GaussStableGeneralLAD} implies that ROP with $L\geq Cm$ random projections from Gaussian distribution is sufficient to ensure the stable recovery of all symmetric rank-1 matrix $X=xx^{T}$ with high probability.
\end{corollary}

\begin{remark}\label{Openproblem-lqRUB}
Note that the condition of Theorem \ref{StableDensityLeastq} is $\ell_q$-RUB for $0<p\leq q\leq 1$. Therefore, if we can show that ROP $\mathcal{A}$ from some distribution satisfies the $\ell_q$-RUB condition for $0<q<1$ with high probability, then we can get a better conclusion.
\end{remark}

\section{Conclusions and Discussion \label{s5}}
\hskip\parindent

In this paper, we consider the matrix recovery from rank-one projection
measurements via nonconvex minimization.

First in Section \ref{s2}, after introducing the $\ell_q$-RUB (Definition \ref{lqRUB}), we consider exact and stable recovery of rank-$r$ matrices from rank-one projection measurements via Schatten-$p$ minimization (\ref{MatrixSp}). And we show that $\ell_q$-RUB condition of order $(k+1)r$ with $C_2/{C_1}<k^{(1/p-1/2)q}$ for some $k>1$ with $kr\in\mathbb{Z}_+$ and $0<p\leq q\leq 1$ is sufficient for exact recovery of all rank-$r$ matrices (Theorem \ref{ROP-lp-Exact}) in Subsection \ref{s2.2}. Subsection \ref{s2.3} considers extensions to the noisy case. We get stable recovery via Schatten-$p$ model (\ref{MatrixSp}) with $\mathcal{B}=\mathcal{B}^{\ell_q}(\eta_1)\cap\mathcal{B}^{DS}(\eta_2)$ (Theorem \ref{ROP-lp-lq-DS}), by combining $\mathcal{B}=\mathcal{B}^{\ell_q}(\eta_1)$ (Proposition \ref{ROP-lp-lq}) with $\mathcal{B}=\mathcal{B}^{DS}(\eta_2)$ (Proposition \ref{ROP-lp-DS}). And our condition is also sufficient for symmetric rank-one projections (Corollary \ref{SROP-lp-Exact} and Corollary \ref{SROP-lp-lq-DS}).

By the proofs of Theorems \ref{ROP-lp-lq} and \ref{ROP-lp-DS}, we also obtain that the robust rank null sapce property of order $r$ can be deduced from the $\ell_q$-RUB of order $(k+1)r$ for some $k>1$ (Theorem \ref{RUB-RNSP}) in Section \ref{s3}.

And in Section \ref{s4}, we consider the stable recovery via the least $q$ minimization (\ref{DensityLeastq}) for $0<q\leq 1$ under $\ell_q$-RUB condition. We show that our condition in Theorem \ref{ROP-lp-Exact} is still sufficient (Theorem \ref{StableDensityLeastq}). And we also consider recovering matrix for ROP $\mathcal{A}$ from Gaussian distribution via least $1$ minimization, and we show that with high probability, ROP with $L\geq Cr(m+n)$ random projections from Gaussian distribution is sufficient to ensure stable recovery of all rank-$r$ matrices (Theorem \ref{GaussStableGeneralLAD}).

However, our $(k+1)r$~$(k>1)$ order RUB condition for $q=1$ (Theorem \ref{ROP-lp-Exact}) is a litter stronger than the $kr$ order RUB condition in \cite{CZ2015}. Therefore, this condition may be improved further (Remark \ref{Openproblem1}). And note that Theorem \ref{StableDensityLeastq} show that $\ell_q$-RUB condition can guarantee the stable recovery via least $q$ minimization (\ref{DensityLeastq}). Therefore, finding a ROP $\mathcal{A}$ from some distribution satisfing the $\ell_q$-RUB condition for $0<q<1$ is one direction of our future research (Remark \ref{Openproblem-lqRUB}).


\textbf{Acknowledgement}: Wengu Chen is supported by National Natural Science Foundation of China (No. 11371183).





\begin{thebibliography}{SIBL}

\vspace{-0.3cm}

\bibitem{AR2015}A. Ahmed and J. Romberg, Compressive multiplexing of correlated signals, IEEE Trans. Inform. Theory, 61(2015), 479-498.

\bibitem{ABHMM2013}P. Alquier, C. Butucea, M. Hebiri, K. Meziani, and T. Morimae,Rank penalized estimation
of a quantum system, Phys. Rev. A (3), 88(2013), 032113-1-032113-9.

\bibitem{AEP2008}A. Argyriou, T. Evgeniou, and M. Pontil, Convex multi-task feature learning, Mach. Learn., 73(2008), 243-272.

\bibitem{A2014}K. M. R. Audenaert, A generalisation of Mirsky¡¯s singular value inequalities, available at http://arxiv.org/abs/1410.4941, 2014.

\bibitem{BK1978}G. Bassett, R. Koenker, Asymptotic theory of least absolute error regression, J. Amer. Statist. Assoc., 73 (1978), 618-621.

\bibitem{BLWY2006} P. Biswas, T.-C. Lian, T.-C. Wang and Y. Ye, Semidefinite programming based algorithms for sensor network localization, ACM Transactions on Sensor Networks, 2(2006), 188-220.

\bibitem{BG2010}I. Borg and P. Groenen, Modern Multidimensional Scaling, New York, NY, USA: Springer, 2010.


\bibitem{CW2011}T. T. Cai and L. Wang, Orthogonal matching pursuit for sparse signal recovery with noise, IEEE Trans. Inform. Theory, 57(2011), 4680-4688.

\bibitem{CMW2013}T. T. Cai, Z. Ma and Y. Wu, Sparse PCA: optimal rates and adaptive estimation, Ann. Statist. 41(2013), 3074-3110.

\bibitem{CMW2015}T. T. Cai, Z. Ma and Y. Wu, Optimal estimation and rank detection for sparse spiked covariance matrices, Probab. Theory Relat. Fields, 161(2015), 781-815.

\bibitem{CZ2014}T. T. Cai and A. Zhang, Sparse representation of a polytope and recovery of sparse signals and low-rank matrices, IEEE Trans. Inform. Theory, 60(2014), 122-132.

\bibitem{CZ2015}T. T. Cai and A. Zhang, ROP: Matrix recovery via rank-one projections, Ann. Statist., 43(2015), 102-138.

\bibitem{CESV2013-2015}E. J. Cand\`{e}s, Y. Eldar, T. Strohmer and V. Voroninski, Phase retrieval via matrix completion, SiAM J. Imaging Sciences, 6(2013), 199-225. And also in SIAM Rev., 57(2015), 225-251.


\bibitem{CL2013}E. J. Cand\`{e}s, X. Li, Solving quadratic equations via PhaseLift when there are about as many equations as unknowns, Found. Comput. Math., 14(2014), 1017-1026.

\bibitem{CP2011}E. J. Cand\`{e}s and Y. Plan, Tight oracle inequalities for low-rank matrix recovery from a minimal number of noisy random measurements, IEEE Trans. Inform. Theory, 57(2009), 2342-2359.

\bibitem{CRT2006}E. J. Cand\`{e}s, J. K. Romberg and T. Tao, Stable signal recovery from incomplete and inaccurate measurements, Comm. Pure Appl. Math., 59(2006), 1207-1223.

\bibitem{CRT2006-1}E. J. Cand\`{e}s, J. K. Romberg and T. Tao, Robust uncertainly principles: Exact signal reconstruction from highly incomplete frequency information, IEEE Trans. Inform. Theory, 52(2006), 489-509.

\bibitem{CSV2013}E. J. Cand\`{e}s, T. Strohmer and V. Voroninski, Phaselift: exact and stable signal recovery from magnitude
measurements via convex programming, Comm. Pure Appl. Math., 66(2013), 1241-1274.

\bibitem{CT2005}E. J. Cand\`{e}s and T. Tao, Decoding by linear programming, IEEE Trans. Inform. Theory, 51(2005), 4203-4215.

\bibitem{CT2007}E. Cand\`{e}s and T. Tao, The Dantzig selector: Statistical estimation when
$p$ is much larger than $n$, Ann. Statist., 35(2007), 2313-2351.

\bibitem{CSV2013}E. J. Cand\`{e}s, T. Strohmer and V. Voroninski, PhaseLift: Exact and stable signal recovery from magnitude measurements via convex programming, Comm. Pure Appl. Math., 66(2013), 1241-1274.

\bibitem{CP2011}A. Chambolle, T. Pock, A first-order primal-dual algorithm for convex problems with
applications to imaging, J. Math. Imaging Vision, 40(2011), 120-145.

\bibitem{CS2008}R. Chartrand and V. Staneva, Restricted isometry properties and nonconvex compressive sensing, Inverse Problems, 24(2008), 035020-1-035020-14.

\bibitem{CL2015}W. Chen and Y. Li, Stable recovery of low-rank matrix via nonconvex Schatten $p$--minimization, Science China Mathematics, 58(2015), 2643-2654.

\bibitem{CCG2015}Y. Chen, Y. Chi and A. J. Goldsmith, Exact and stable covariance estimation from quadratic sampling via convex programming, IEEE Trans. Inform. Theory,  61(2015), 4034-4059.


\bibitem{CP2011}P. Combettes and J.-C. Pesquet, Proximal splitting methods in signal processing, In Fixed-
Point Algorithms for Inverse Problems in Science and Engineering, ed. by H. Bauschke,
R. Burachik, P. Combettes, V. Elser, D. Luke, H. Wolkowicz (Springer, New York, 2011), 185-212.


\bibitem{DE2012}M. E. Davies and Y. C. Eldar, Rank awareness in joint sparse recovery, IEEE Trans. Inform. Theory, 58(2012), 1135-1146.

\bibitem{D2006}D. L. Donoho, Compressed sensing, IEEE Trans. Inform. Theory, 52(2006), 1289-1306.


\bibitem{DET2006}D. L. Donoho, M. Elad, and V. N. Temlyakov, Stable recovery of sparse overcomplete representations in the presence of noise, IEEE Trans. Inform. Theory, 52(2006), 6-18.

\bibitem{D2004}S. T. Dumais, Latent semantic analysis,  Ann. Rev. Inf. Sci. Tech., 38(1)(2004), 188-230.


\bibitem{FR2013}S. Foucart and H. Rauhut, A mathematical introduction to compressive sensing, Applied and
Numerical Harmonic Analysis Series, New York: Birkh$\ddot{a}$user/Springer, 2013.

\bibitem{GLFBE2010}D. Gross,. Y.-K. Liu, S. T. Flammia, S. Becker and J. Eisert, Quantum state
tomography via compressed sensing, Phys. Rev. Lett., 105(2010), 150401-150404.

\bibitem{KKRT2016}M. Kabanava, R. Kueng, H. Rauhut and U. Terstiege, Stable low-rank matrix recovery via null space properties, Information and Inference, 5(2016), 405-441.

\bibitem{KX2013}L. Kong and N. Xiu, Exact low-rank matrix recovery via nonconvex Schatten $p$-minimization, Asia-Pac. J. Oper. Res., 30(2013), 1340010-1-1340010-13.

\bibitem{LB2010}K. Lee and Y. Bresler, ADMiRA: Atomic decomposition for minimum rank approximation, IEEE
Trans. Image Process., 56(2010), 4402-4416.

\bibitem{LV2013}X. Li and V. Voroninski, Sparse signal recovery from quadratic measurements via convex programming, SIAM J. Math. Anal., 45(2012), 3019-3033.

\bibitem{LSC2017}Y. Li, Y. Sun and Y. Chi, Low-rank positive semidefinite matrix recovery from corrupted rank-One measurements, Trans. Signal Process., 65(2017), 397-408.


\bibitem{LLR1995}N. Linial, E. London, and Y. Rabinovich, The geometry of graphs and some of its algorithmic applications, Combinatorica, 15(1995), 215-245.

\bibitem{LV2009} Z. Liu and L. Vandenberghe, Interior-point method for nuclear norm approximation with application to system identification, SIAM J. Matrix Anal. Appl., 31(2009), 1235-1256.

\bibitem{OTJ2010}G. Obozinski, B. Taskar, and M. I. Jordan, Joint covariate selection and joint subspace selection for multiple classification problems, Stat. Comput., 20(2010),  231-252.

\bibitem{P1984}J. L. Powell, Least absolute deviations estimation for the censored regression model, J. Econometrics, 25(1984), 303-325.

\bibitem{RFP2010}B. Recht, M. Fazel, and P. A. Parrilo, Guaranteed minimum-rank solutions of linear matrix
equations via nuclear norm minimization, SIAM Rev., 52(2010), 471-501.

\bibitem{SY2007}A M.-C. So and Y. Ye, Theory of semidefinite programming for sensor network localization,  Math. Program. Ser. B, 109(2007), 367-384.

\bibitem{S2011}Q. Sun, Sparse approximation property and stable recovery of sparse signals from noisy measurements, IEEE Trans. Signal Process., 59(2011), 5086-5090.

\bibitem{TW2013}J. Tanner and K. Wei, Normalized iterative hard thresholding for matrix completion, SIAM J. Sci.
Comput., 59(2013), 7491-7508.

\bibitem{T2000}M. W. Trosset, Distance matrix completion by numerical optimization, Comput. Optim. Appl., 17(2000), 11-22.

\bibitem{WKT2017}H. Wang, L. Kong and J. Tao, The linearized alternating direction method of multipliers for sparse group LAD model, Optim Lett, 1(2017), 1-21.

\bibitem{W2012}L. Wang, The $L_1$ penalized LAD estimator for high dimensional linear regression, J. Multivariate Anal., 120(2012), 135-151.


\bibitem{W2015}Z. Wei, Nuclear norm penalized LAD estimator for low rank matrix recovery, Ph. D., Massachusetts Institute of Technology, Cambridge, USA, 2015.

\bibitem{YS2016}M.-C. Yue and A. M.-C. So, A perturbation inequality for concave functions of singular values and its applications in low-rank matrix recovery, Appl. Comput. Harmon. Anal., 40(2016), 396-416.

\bibitem{ZHZ2013} M.  Zhang, Z.-H. Huang and Y. Zhang, Restricted $p$ -isometry properties of nonconvex matrix recovery, IEEE Trans. Inform. Theory, 59(2013), 4316-4323.

\bibitem{ZL2017}R. Zhang and S. Li, A proof of conjecture on restricted isometry property constants $\delta_{tk}~(0<t<\frac{4}{3})$, IEEE Trans. Inform. Theory, (2017), DOI: 10.1109/TIT.2017.2705741.

\bibitem{ZL2017-1}R. Zhang and S. Li, Optimal RIP bounds for sparse signals recovery via $\ell_p$ minimization, Appl. Comput. Harmon. Anal., 2017, doi:10.1016/j.acha.2017.10.004.


\end{thebibliography}
\end{document}